\documentclass{article}
\usepackage[english]{babel}
\usepackage{amsmath,amssymb,graphicx,xcolor,latexsym,theorem}
\usepackage{comment}
\usepackage{booktabs}
\newcommand{\ra}[1]{\renewcommand{\arraystretch}{#1}}
\newcommand{\assign}{:=}

\newcommand{\infixand}{\text{ and }}
\newcommand{\nobracket}{}
\newcommand{\R}{\mathbb{R}}

\newcommand{\N}{\mathbb{N}}

\newcommand{\geoRP}[1]{\mathbf{#1}}

\newcommand{\Sig}[1]{\mathbf{#1}^{<\infty}}

\newcommand{\tmmathbf}[1]{\ensuremath{\boldsymbol{#1}}}
\newcommand{\tmop}[1]{\ensuremath{\operatorname{#1}}}

\newcommand{\vertiii}[1]{{\left\vert\kern-0.25ex\left\vert\kern-0.25ex\left\vert #1 
    \right\vert\kern-0.25ex\right\vert\kern-0.25ex\right\vert}}
\newenvironment{proof}{\noindent\textbf{Proof\ }}{\hspace*{\fill}$\Box$\medskip}

{\theorembodyfont{\rmfamily}\newtheorem{example}{Example}}

\newtheorem{theorem}{Theorem}[section]
\newtheorem{remark}[theorem]{Remark}
\newtheorem{lemma}[theorem]{Lemma}
\newtheorem{proposition}[theorem]{Proposition}
\newtheorem{Assumption}[theorem]{Assumption}
\newtheorem{corollary}[theorem]{Corollary}
\newtheorem{definition}[theorem]{Definition}
\theoremstyle{break}

\usepackage[colorlinks,citecolor=red,urlcolor=blue,bookmarks=false,hypertexnames=true]{hyperref}
\usepackage{shuffle}
\usepackage{mathrsfs}
\usepackage{mathtools}
\usepackage{atbegshi}
\usepackage[colorlinks]{hyperref}
\usepackage[nameinlink,capitalize]{cleveref} 
\usepackage[autostyle]{csquotes}
\usepackage[toc,page]{appendix}
\DeclareMathOperator*{\esssup}{ess\,sup}
\def\dsqcup{\sqcup\mathchoice{\mkern-7mu}{\mkern-7mu}{\mkern-3.2mu}{\mkern-3.8mu}\sqcup}

\usepackage{todonotes}

\begin{document}
\title{Primal and dual optimal stopping with signatures\thanks{Acknowledgements: All authors gratefully acknowledge funding by the Deutsche Forschungsgemeinschaft (DFG, German Research Foundation) under Germany's Excellence Strategy – The Berlin Mathematics Research Center MATH+ (EXC-2046/1, project ID: 390685689). The authors would like to thank C. Cuchiero and S. Breneis for valuable remarks and helpful discussions about the global approximation in Section \ref{sec:globalapproxi}.}}

\author{Christian Bayer\thanks{Weierstrass Institut (WIAS), Berlin, Germany, email: bayerc@wias-berlin.de.} \quad Luca Pelizzari\thanks{Technische Universität Berlin and Weierstrass Institut (WIAS), Berlin, Germany, email: pelizzari@wias-berlin.de} \quad John Schoenmakers\thanks{Weierstrass Institut (WIAS), Berlin, Germany, email: schoenma@wias-berlin.de.}}

\maketitle

\begin{abstract}
We propose two signature-based methods to solve the optimal stopping problem -- that is, to price American options -- in non-Markovian frameworks. Both methods rely
on a global approximation result for $L^p-$functionals on rough path-spaces,
using linear functionals of robust, rough path signatures. In the primal formulation, we present a non-Markovian generalization of the famous Longstaff-Schwartz algorithm, using linear
functionals of the signature as regression basis. For the dual formulation, we parametrize the space of square-integrable martingales using linear functionals of the signature, and apply a sample average
approximation. We prove convergence for both methods and present first
numerical examples in non-Markovian and non-semimartingale regimes.
\end{abstract}

\section{Introduction}

Stochastic processes with memory play a more and more important role in the modelling of financial markets. In the modelling of equity markets, \emph{rough stochastic volatility models} are now part of the standard toolbox, see, e.g., Gatheral et al. \cite{gatheral2018volatility} and Bayer et al. \cite{bayer2016pricing}. In the same area, \emph{path-dependent stochastic volatility models} Guyon et al. \cite{guyon2023volatility} are a very powerful alternative for capturing memory-effects.  Processes with memory are also an essential tool for modelling the micro-structure of financial markets, driven by the market practice of splitting large orders in many medium size ones, as well as by the reaction of algorithmic traders to such orders. Seen from outside, this materializes as self-excitation of the order flow, and, consequently, \emph{Hawkes processes} are a fundamental tool for modelling order flows, see, e.g., Bouchaud et al. \cite{bouchaud2018trades}.
Beyond finance, processes with memory play an important role in the modelling of many natural phenomena (e.g., earthquakes, see Ogata et al. \cite{ogata1988statistical}) or social phenomena. 

In this paper we study optimal stopping problems in non-Markovian frameworks, that is the underlying price is possibly a stochastic process with memory. For concreteness' sake, let us introduce two processes determining the optimal stopping problem: an underlying \emph{state-process} $X$, together with its natural filtration $\mathbb{F}^X$, and a \emph{reward-process} $Z$, which is $\mathbb{F}^X-$adapted -- think about $X = (S,v)$ for a stock price process $S$ driven by a stochastic variance process $v$ and $Z_t = \phi(t,S_t)$. The optimal stopping problem then consists of solving the following optimization problem  \begin{equation}\label{eq:stoppingintro}
y_0=\sup_{\tau \in \mathcal{S}_0}E[Z_{\tau}],
\end{equation} where $\mathcal{S}_0$ denotes the set of $\mathbb{F}^X-$stopping times on $[0,T]$, for some $T>0$. We merely assume $\alpha$-Hölder continuity for $X$ in our framework, see Section \ref{sec:framework} below, in particular allowing non-Markovian and non-semimartingale state-processes $X$.

The lack of Markov property leads to severe theoretical and computational challenges in the context of optimal control problems, and thus in particular in the optimal stopping problem \eqref{eq:stoppingintro}. Indeed, the primary analytical and numerical framework for stochastic optimal control problems arguably is the associated Hamilton--Jacobi--Bellman (HJB) PDE, in the context of optimal stopping so-called \emph{free-boundary problems}, see Peskir and Shiryaev \cite{peskir2006optimal}. When the state process is not a Markov process, such PDEs do not exist a priori. As noted above, infinite-dimensional (BS)PDE formulations can be given, see, for instance, Bayer et. al \cite{BQY22} for a BSPDE-description of the American option price in rough volatility models. When the underlying dynamics is of stochastic Volterra type, \emph{path-dependent} HJB PDEs could probably be derived following the approach of Bonesini and Jacquier \cite{bonesini2023}. However, most \emph{numerical approximation methods} crucially rely on the Markov property, as well. 

It should be noted that, at least intuitively, all processes with memory can be turned into Markov processes by adding the history to the current state -- but see, e.g., Carmona and Coutin \cite{carmona1998fractional} for a more sophisticated approach in the case of fractional Brownian motion. Hence, theoretical and numerical methods from the Markovian world are, in principle, available, but at the cost of having to work in infinite-dimensional (often very carefully drafted, see. e.g., Cuchiero and Teichmann \cite{cuchiero2020generalized}) state spaces. On the other hand, \emph{Markovian approximations}, i.e., finite-dimensional Markov processes closely mimicking the process with memory, can sometimes be a very efficient surrogate model, especially when high accuracy is achievable with low-dimensional Markovian approximations, see, e.g., Bayer and Breneis \cite{bayer2023efficient}.

Inspired by many successful uses in machine learning (for time-series data), Kalsi et al. \cite{kalsi2020optimal} introduced a model-free method for numerically 
solving a stochastic optimal execution problem. The method is based on the \emph{path signature}, see, e.g., Friz and Victoir \cite{friz2010multidimensional}, and is applicable in non-Markovian settings. This approach was extended to optimal stopping problems in Bayer et al. \cite{bayer2021optimal}, where stopping times were parameterized as first hitting times of affine hyperplanes in the signature-space. A rigorous mathematical analysis of that method was performed and numerical examples verifying its efficiency were provided.

The signature $\mathbf{X}^{<\infty}$ of a path $X:[0,T] \rightarrow \mathbb{R}^d$, is given (at least formally) as the infinite collection of \emph{iterated integrals}, that is for $0\leq t \leq s \leq T$ $$\mathbf{X}^{<\infty}_{s,t} = \bigg \{ \int_s^t\int_s^{t_k}\dots \int_s^{t_2}dX_{t_1}^{i_1}\cdots dX_{t_k}^{i_k}: i_1,\dots,i_k \in \{1,\dots,d\}, k \geq 0 \bigg\}.$$ The signature characterizes the history of the corresponding trajectory, and, hence, provides a systematic way of ``lifting'' a process with memory to a Markov process by adding the past to the state. Relying only on minimal regularity assumptions, the corresponding encoding is efficient, and has nice algebraic properties. In many ways, (linear functionals of) the path signature behaves like an analogue of polynomials on path space, and can be seen as a canonical choice of basis functions on path space. For example, a \emph{Stone-Weierstrass} type of result shows that, restricted to compacts, continuous functionals on path spaces can be approximated by linear functionals of the signature, that is by linear combinations of iterated integrals, see for instance Kalsi et al. \cite[Lemma 3.4]{kalsi2020optimal}.  

As a first contribution, in Section \ref{sec:globalapproxi} we provide an abstract approximation result on $\alpha$-Hölder rough-path spaces, by linear functionals of the \emph{robust} signature, with respect to the $L^p$-norm, see Theorem \ref{mainresultappro} below. As a direct consequence, and under very mild assumptions, we can show that for any $\mathbb{F}^X-$progressive process $(\xi_t)_{t\in [0,T]}$, we can find a sequence $(l_n)_{n\in \mathbb{N}}$ of linear functionals on the state-space of the signature, such that \begin{equation}\label{eq:convergenceintro}
\lim_{n\to \infty }E\bigg [ \int_0^T(\xi_t-\langle \mathbf{X}^{<\infty}_{0,t},l_n\rangle)^2dt\bigg ]= 0,
\end{equation} see Corollary \ref{mainresultcoro}  below for the details. This result is in marked contrast to the standard universal approximation result for signatures as usually formulated, which only provides uniform convergence on compact subsets of the path space. Let us mention two related works on global approximation with signatures: First, in Cuchiero et al. \cite{cuchiero2023global} the authors study global signature approximations based on a version of the Stone-Weierstrass result for so-called \emph{weighted spaces}, see \cite[Theorem 3.6]{cuchiero2023global}. Compared to our theory, while they do not require a robust version of the signature, such weighted spaces need to be crafted carefully. Secondly, in the recent work Alaifari and Schell \cite{alaifari2023conditioning}, the authors obtain (independently from us) a similar global \(L^p\) approximation result for robust signatures of bounded variation paths, see \cite[Proposition 4.5]{alaifari2023conditioning}, but based on a Monotone Class rather than a Stone-Weierstrass argument.

Returning to the optimal stopping problem \eqref{eq:stoppingintro}, in Section \ref{sec: OptimalStoppingSection} we generalize two standard techniques from the Markovian case to the non-Markovian case using signatures, namely the \emph{Longstaff--Schwartz algorithm} \cite{longstaff2001valuing} and Rogers' \emph{dual martingale method} \cite{rogers2002monte}. Denoting by $Y$ the \emph{Snell envelope} to the optimal stopping problem, see below for more details, the Longstaff--Schwartz algorithm is based on the \emph{dynamic programming principle}, that is %, (for a discrete-time problem or a discretized continuous time problem),
\begin{equation*}
    Y_t = \max \big( Z_t, \, E[Y_{t + \Delta t}| \mathcal{F}^X_t] \big).
\end{equation*}
If $X$ is a Markov process, then $E[Y_{t + \Delta t} |\mathcal{F}^X_t] = E[Y_{t + \Delta t} | X_t]$, which can be efficiently computed using regression (\emph{least-squares Monte Carlo}). In the non-Markovian case, an application of the global approximation result Theorem \ref{mainresultappro}, i.e. the convergence in \eqref{eq:convergenceintro}, shows that (under minimal assumptions) a Longstaff--Schwartz algorithm converges when the conditional expectation is approximated by linear functionals of the signature, that is
\[
    t \mapsto E[Y_{t + \Delta t}|\mathcal{F}^X_t] \approx \langle \mathbf{X}^{<\infty}_{0,t},\ell \rangle,
\]
see Proposition~\ref{THM:LSconvergence1}. \\ \\
Regarding the dual method, we rely on Rogers' characterization that 
\[
    y_0 = \inf_{M \in \mathcal{M}_0^2} E\Big[ \sup_{t \in [0,T]} (Z_t - M_t)\Big],
\]
where the $\inf$ is taken over all square-integrable martingales $M$ starting at $0$. If the underlying filtration is Brownian, such martingales can be written as stochastic integrals w.r.t.~a Brownian motion $W$, that is $M_t = \int_0^t \xi_s dW_s$ for some $\mathbb{F}^X-$progressive process $\xi$. The approximation result Theorem \ref{mainresultappro}, i.e. the convergence in \eqref{eq:convergenceintro}, suggests to approximate the integrand by linear functionals of the signature
\[
    t \mapsto \xi_t \approx \langle \mathbf{X}^{<\infty}_{0,t},\ell \rangle,
\]
and we prove convergence after taking the infimum over all linear functionals $l$, that is 
\begin{equation}y_0 = \inf_\ell E\bigg[ \sup_{t \in [0,T]} \bigg (Z_t - \int_0^t\langle \mathbf{X}^{<\infty}_{0,s},\ell \rangle dW_s\bigg )\bigg], \label{SA_LMC}
\end{equation}
see Proposition~\ref{propositionapproxi}. 
For numerically solving the dual problem \eqref{SA_LMC}  we carry out a \emph{Sample Average Approximation} (SAA) with respect to the  coefficients of 
the linear functional of the signature. For a Markovian environment, a related SAA procedure was earlier proposed in Desai et al. \cite{desai2012pathwise} and recently refined in Belomestny and Schoenmakers \cite{belomestny2023from} and  Belomestny et al. \cite{belomestny2022solving}
using a suitable  randomization. 
An important feature of the SAA method is that it relies on nonnested Monte Carlo simulation and  thus is very fast in comparison to the classical nested Monte Carlo method
by Andersen and Broadie \cite{andersen2004primal}.  

For both the Longstaff--Schwartz and the dual signature methods, we also prove convergence of the finite sample approximations when the number of samples grows to infinity, see Proposition \ref{THM:LSconvergence2} and Proposition \ref{THM:SAAapproxi2}. It is worth to notice that, after independent resimulations, the Longstaff--Schwartz algorithm  yields lower-biased, whereas the dual method gives upper-biased values to the optimal stopping problem \eqref{eq:stoppingintro}, and thus applying both methods produces confidence intervals for the true value of $y_0$. 

Finally, in Section \ref{sec:Numerics} we provide first numerical examples based on the primal and dual signature-based approaches, in two non-Markovian frameworks: First, in Section \ref{sec: fBM}, we study the task of optimally stopping \emph{fractional Brownian motion} for a wide range of Hurst-parameter $H\in (0,1)$, representing the canonical choice of a state-process leaving the Markov regime. The same problem was already studied in Becker et al. \cite{becker2019deep}, and later in Bayer et al.  \cite{bayer2021optimal}, and we compare our lower, resp. upper bounds with the results therein. Secondly, in Section \ref{sec:AmericanOptions} we consider the problem of computing American options prices in the \emph{rough Bergomi model} Bayer et al. \cite{bayer2016pricing}, and we compare our price intervals with Bayer et al. \cite{bayer2020pricing}, resp. Goudenege et al. \cite{goudenege2020machine}, where lower-bounds were computed in the same model.

\subsection{Notation}\label{presectionRP} 

For $d,K\in \mathbb{N}$ we define the so-called \emph{extended tensor-algebra}, and the $K$-step truncation thereof by $$T((\R^d))= \prod_{k\geq 0}(\R^d)^{\otimes k}, \quad T^{\leq K}(\R^d)= \prod_{k=0}^K(\R^d)^{\otimes k},$$ where we use the convention $(\R^d)^{\otimes 0} = \mathbb{R}$. For more details, including natural operations such as sum $+$ and product $\star$ on these spaces, see for instance Friz and Victoir \cite[Section 7.2.1]{friz2010multidimensional}. For any word $w=i_1\cdots i_n$ for some $n\in \N$ with letters $i_1,\dots,i_n \in \{1,\dots,d\}$, we define the \emph{degree} of $w$ as the length of the word, that is $\mathrm{deg}(w)=n$, and denote by $\emptyset$ the empty word with $\mathrm{deg}(\emptyset)=0$. Moreover, for $\mathbf{a}\in T((\R^d))$, we denote by $\langle \mathbf{a},w \rangle$ the element of $\mathbf{a}^{(n)} \in (\mathbb{R}^d)^{\otimes n}$ corresponding to the basis element $e_{i_1}\otimes \cdots \otimes e_{i_n}$. Denoting by $\mathcal{W}^d$ the linear span of words, the pairing above can be extended linearly $\langle \cdot,\cdot \rangle: T((\R^d))\times \mathcal{W}^d \rightarrow \R$. For an element $\ell\in \mathcal{W}^{d}$, that is $\ell= \lambda_1w_1+\dots+\lambda_nw_n$ for some words $w_1,\dots,w_n$ and scalars $\lambda_1,\dots,\lambda_n \in \mathbb{R}$, we define the degree of $\ell$ by $\mathrm{deg}(\ell):= \max_{1\leq i \leq n}\mathrm{deg}(w_i)$, and for $K\in \mathbb{N}$ we denote by $\mathcal{W}^d_{\leq K} \subset \mathcal{W}^d$ the subset of elements $l$ with $\mathrm{deg}(\ell) \leq K$. For two words $w$ and $v$ we denote by $\shuffle$ the \emph{shuffle product} \begin{equation}\label{def:shuffle-product}
w \shuffle \emptyset = \emptyset \shuffle w = w, \quad wi \shuffle vj :=  (w \shuffle vj )i +  (wi \shuffle v )j, \quad i,j\in \{1,\dots,d\},
\end{equation} which bi-linearly extends to the span of words $\mathcal{W}^d$. We further define the \textit{free nilpotent Lie group} over $\R^d$ by \begin{equation*}
G((\R^d))= \big \{\mathbf{a}\in T((\R^d)) \setminus{ \{\mathbf{0}\}}: \langle \mathbf{a},w \rangle \langle \mathbf{a},v \rangle = \langle \mathbf{a},w \shuffle v \rangle, \forall w,v \in \mathcal{W}^d \big \},
\end{equation*} see \cite[Chapter 7.5]{friz2010multidimensional} for details.

For $\alpha \in (0,1)$ we denote by $C^{\alpha}([0,T],\R^d)$ the space of $\alpha$-Hölder continuous paths $X$, that is $X:[0,T] \rightarrow \R^d$ such that \begin{equation*}
    \Vert X \Vert_{\alpha;[0,T]} = \sup_{0 \leq s< t \leq T}\frac{\Vert X_t-X_s \Vert}{|t-s|^{\alpha}} < \infty,
\end{equation*} where $\Vert \cdot \Vert$ denotes the Euclidean norm on $\R^d$. Denote by $\Delta_{[0,T]}^2$ the standard simplex $\Delta_{[0,T]}^2:=  \{ (s,t) \in [0,T]^2: 0 \leq s \leq t \leq T  \}.$ For $L \in \mathbb{N}$ and any two-parameter function on the truncated tensor-algebra \begin{equation*}
\Delta_{[0,T]}^2 \ni (s,t) \mapsto \mathbf{X}_{s,t}= (1,\mathbf{X}_{s,t}^{(1)},\dots,\mathbf{X}_{s,t}^{(L)}) \in T^{\leq L}(\R^d), 
\end{equation*} we denote by $\vertiii{\cdot}_{(\alpha,L)}$ the norm given by \begin{equation*}
 \vertiii{\mathbf{X}}_{(\alpha,L)}:=\max_{1\leq l \leq L}\bigg (\sup_{0\leq s < t \leq T}\frac{\Vert \mathbf{X}_{s,t}^{(l)} \Vert}{|t-s|^{l\alpha }} \bigg )^{1/l}
\end{equation*} We denote by $\mathscr{C}^{\alpha}_g([0,T],\R^d)$ the space of geometric $\alpha-$Hölder rough paths $\mathbf{X}$ on $\R^d$, which is the $\vertiii{\cdot}_{(\alpha,L)}-$closure of $L$-step signatures of Lipschitz continuous paths $X:[0,T] \rightarrow \R^d$ for $L= \lfloor 1/\alpha \rfloor$. More precisely, for every $\mathbf{X}\in \mathscr{C}^{\alpha}_g$ there exists a sequence $(X^n)_{n\in \mathbb{N}}\subset \mathrm{Lip}([0,T],\R^d)$ such that $\vertiii{\mathbf{X}^n-\mathbf{X}}_{(\alpha,L)} \rightarrow 0$ as $n\to \infty$, where $\mathbf{X}^n$ is the $L$-step signature of $X^n$, that is $$\mathbf{X}_{s,t}^n:= \bigg (\int_{s<t_1<\dots <t_l<t}\otimes dX^n_{t_1} \dots \otimes dX^n_{t_l}: 0 \leq l \leq L \bigg) \in G^{\leq L}(\R^d),$$ where the integrals are defined in a Riemann-Stieljes sense. For the rest of this paper, we will always fix $L=\lfloor 1/\alpha \rfloor$ and use the shorter notation $\vertiii{\cdot}_{\alpha}:= \vertiii{\cdot}_{(\alpha,\lfloor 1/\alpha \rfloor)}$. For any $\mathbf{X}\in \mathscr{C}^{\alpha}_g$ we denote by $\mathbf{X}^{<\infty}$ the \emph{rough-path signature}, which is the unique (up to tree-like equivalence $\sim_t$, see Boedihardjo et al. \cite{boedihardjo2016signature} for details) path from Lyons' extension theorem \cite[Theorem 3.7]{lyons1998differential}, that is \begin{equation*}
\Delta_{[s,t]}^2:[0,T] \ni (s,t) \mapsto \mathbf{X}^{<\infty}_{s,t}= (1,\mathbf{X}_{s,t}^{(1)},\dots, \mathbf{X}_{s,t}^{(L)},\mathbf{X}_{s,t}^{(L+1)},\dots) \in G(\R^d),
\end{equation*} such that $$\Vert \mathbf{X}^{(k)}\Vert_{k\alpha} < \infty \quad \forall k \geq 0, \quad \mathbf{X}^{<\infty}_{s,t} =\mathbf{X}^{<\infty}_{s,u} \star \mathbf{X}^{<\infty}_{u,t} \quad s\leq u \leq t,$$ where the latter is called \textit{Chen's} relation. Finally, by considering time-augmented paths $\widehat{X}_t = (t,X_t)$, and their geometric rough-path lifts $\widehat{\mathbf{X}}$, the signature maps becomes unique due to the strictly monotone time component. We denote by $\widehat{\mathscr{C}}^{\alpha}_g([0,T],\mathbb{R}^{d+1})$ the space of geometric $\alpha$-Hölder rough-path lifts of $\widehat{X}_t = (t,X_t)$, $X\in {C}^{\alpha}([0,T],\R^d)$. We will often use the shorter notation $\widehat{\mathscr{C}}^{\alpha}_g$ when it is clear from the context that we are working on the fixed time-interval $[0,T]$.

\section{Global approximation with rough-path signatures}\label{sec:globalapproxi}
In this section we present the theoretical foundation of this paper, which consists of a global approximation result based on \emph{robust} rough-path signatures. 

\subsection{The space of stopped rough paths}

For $\alpha \in (0,1)$, we consider an $\alpha$-Hölder continuous path $X:[0,T] \rightarrow \mathbb{R}^d$ starting at $X_0=x_0 \in \mathbb{R}^d$, and denote by $\geoRP{X} $ a geometric rough-path lift of the time-augmentation $(t,X_t)$, that is $\geoRP{X} \in \widehat{\mathscr{C}}^{\alpha}_g([0,T],\R^{d+1})$. 

\begin{definition}\label{def:stoppedRP} For any $\alpha \in (0,1)$ and $T>0$, the space of stopped $\widehat{\mathscr{C}}^{\alpha}_g-$paths is defined by the disjoint union 
    $$\Lambda^{\alpha}_T:= \bigcup_{t\in [0,T]}\widehat{\mathscr{C}}^{\alpha}_g([0,t],\R^{d+1}).$$ Moreover, we equip the space $\Lambda_T^{\alpha}$ with the final topology %\footnote{Recall that for a topological space $Y$ and $f:Y \to X$, the final topology on $X$, induced by $f$, consists of all sets $A\subseteq X$ s.t. $f^{-1}(A)$ is open.}
    induced by the map $$\phi:[0,T] \times \widehat{\mathscr{C}}^{\alpha}_g([0,T],\R^{d+1}) \longrightarrow \Lambda_T^{\alpha},\quad \phi(t,\mathbf{x}) = \mathbf{x}|_{[0,t]}.$$
\end{definition} 
The reason to work on this space is the following: If $X$ is a stochastic process, and $\mathbf{X}$ denotes the random geometric lift of $(t,X_t)$, we define $\mathcal{F}_t^{\mathbf{X}}=\sigma(\mathbf{X}_{0,s}:s\leq t)$ for $0\leq t \leq T$, i.e. the natural filtration generated by $\mathbf{X}$. In Lemma \ref{progressivLemma} below, we show that any $\mathbb{F}^{\mathbf{X}}-$progressive process $(A_t)_{t\in [0,T]}$ can be expressed as $A_t=f(\mathbf{X}|_{[0,t]})$, where $f$ is a measurable function on $\Lambda_T^{\alpha}$. Thus, progressively measurable processes can be thought of measurable functionals on $\Lambda_T^{\alpha}$, and we will discuss approximation results for the latter below. Similar spaces have already been considered in relation with \textit{functional Itô calculus} in Cont and Fournié \cite{cont2013functional} and Dupire \cite{dupire2019functional}, and for $p$-rough paths in Kalsi et al. \cite{kalsi2020optimal}, and more recently in relation with optimal stopping in Bayer et al. \cite{bayer2021optimal}.\begin{remark}\label{rmk:polish} One can also introduce a metric $d_{\Lambda}$ on the space $\Lambda^{\alpha}_T$. Let $\mathbf{y}\in \Lambda_T^{\alpha}$, that is there exists $0\leq s \leq T$ such that $\mathbf{y}=\mathbf{y}|_{[0,s]}\in \widehat{\mathscr{C}}^{\alpha}_g([0,s])$. Now for any $t\geq s$, we can extend $\mathbf{y}|_{[0,s]}$ to a geometric rough path $\tilde{\mathbf{y}}|_{[0,t]}\in \widehat{\mathscr{C}}^{\alpha}_g([0,t])$ as follows: By geometricity, there exists a sequence of smooth path $(u,y^n_u)_{u\in [0,s]}$, such that the (canonical) rough-path lifts converge to $\mathbf{y}|_{[0,s]}$. Then, we define $\tilde{\mathbf{y}}|_{[0,t]}$ to be the rough-path limit of the canonical lift of $u \mapsto (u,y^n_{u \land s})$ for $0\leq u \leq t$. This construction can be used to define \begin{equation*}
    d_{\Lambda}(\mathbf{x}|_{[0,t]},\mathbf{y}|_{[0,s]}) := \vertiii{\mathbf{x}-\tilde{\mathbf{y}}} _{\alpha:[0,t]}+ |t-s|, \quad s \leq t.
\end{equation*} It has been proved in \cite[Lemma A.1]{bayer2021optimal}, in the case of $p$-rough paths, the topology of the metric space $(\Lambda^{p}_T,d_{\Lambda})$ coincides with the final topology, and the space of stopped geometric rough paths is Polish. A similar argument can be done for the $\alpha$-Hölder case, by replacing the $p$-variation norm by the $\alpha$-Hölder norm and using the fact that $\widehat{\mathscr{C}}^{\alpha}_g$ is Polish, see Friz and Victoir \cite[Proposition 8.27]{friz2010multidimensional}.
\end{remark} 

\begin{remark}\label{rmk:geometricRV}
    Let $\mathbf{X}$ be a stochastic process on $\widehat{\mathscr{C}}^{\alpha}_g([0,T],\mathbb{R}^{d+1})$. It is discussed in \cite[Appendix A.1]{friz2010multidimensional} that $\mathbf{X}$ can be regarded as a random variable in $\widehat{\mathscr{C}}^{\alpha}_g([0,T],\mathbb{R}^{d+1})$, and its law $\mu_{\mathbf{X}}$ is a Borel measure on the Borel $\sigma$-algebra $\mathcal{B}^{\alpha}$ with respect to $\vertiii{\cdot}_{\alpha}$. Moreover, we define the product measure $d\mu := dt \otimes d\mu_{\mathbf{X}}$. For the surjective map $\phi: [0,T] \times \widehat{\mathscr{C}}_g^{\alpha}([0,T],\mathbb{R}^{d+1}) \rightarrow \Lambda_T^{\alpha}$ defined above, we can define the pushforward measure $\widehat{\mu}$ on $\Lambda^{\alpha}_T$, in symbols $\widehat{\mu} := \phi \# \mu$, which is given by \begin{equation*}
    \widehat{\mu}(A):= \mu \big (\phi^{-1}(A)\big ) \text{ for all } A\in \mathcal{B}(\Lambda^{\alpha}_T).
\end{equation*}
\end{remark}  Consider the space $\mathbb{H}^2$ of $\mathbb{F}^{\mathbf{X}}-$progressive processes $A$, such that \begin{equation}\label{def:prognorm}
    \Vert A\Vert_{\mathbb{H}^2}^2:=E\bigg [\int_0^TA_s^2ds\bigg ]<\infty.
\end{equation} The following result justifies the consideration of the space $\Lambda_T^{\alpha}$.

\begin{lemma}\label{progressivLemma}
     For any process $A\in \mathbb{H}^2$ and $ \alpha \in (0,1)$, there exists a measurable function $f:(\Lambda^{\alpha}_T,\mathcal{B}(\Lambda^{\alpha}_T)) \longrightarrow (\mathbb{R},\mathcal{B}(\mathbb{R}))$, such that $A_t = f(\mathbf{X}|_{[0,t]})$ almost everywhere.
\end{lemma} 

\begin{proof}
    Consider the space of elementary, $\mathbb{F}^{\mathbf{X}}-$progressive processes, that is processes of the form 
    \begin{equation}
        A^n_t(\omega):= \xi^n_0(\omega)1_{\{0\}}(t) + \sum_{j=1}^{m_n-1}1_{(t_j^n,t_{j+1}^n]}(t)\xi^n_{t_{j}}(\omega),\label{elementary}
    \end{equation}
    where $0\leq t^n_0 < \dots < t^n_{m_n}\leq T$, and $\xi^n_{t_j}$ is a $(\mathcal{F}^{\mathbf{X}}_{t_j})-$measurable, square integrable random variable. A standard result for the construction of stochastic integrals, shows that this space is dense in $\mathbb{H}^2$, this can be found in Karatzas and Shreve \cite[Lemma 3.2.4]{karatzas1991brownian} for instance. Thus, we can find $A^n$ of the form \eqref{elementary}, such that $A^n \longrightarrow A$ for almost every $(t,\omega)$. Since the random variable $\xi_{t_j}^n$ is measurable with respect to the $\sigma$-algebra $\mathcal{F}^{\mathbf{X}}_{t_j}:=\sigma(\mathbf{X}_{0,s}:s\leq t_j)=\sigma(\mathbf{X}|_{[0,t_j]})$, by the Doob-Dynkin Lemma there exists a Borel measurable function $F^n_j:\widehat{\mathscr{C}}^{\alpha}_g([0,t_j],\mathbb{R}^{d+1}) \longrightarrow \mathbb{R}$ such that $\xi^n_{t_j}(\omega)= F_j^n(\mathbf{X}|_{[0,t_j]}(\omega))$. Then the functions 
    \begin{equation*}
        [0,T]\times \widehat{\mathscr{C}}^{\alpha}_g \ni (t,\mathbf{x}) \mapsto 1_{(t_j^n,t_{j+1}^n]}(t)F_j^n(\mathbf{x}|_{[0,t_j]})
    \end{equation*} are $ (\mathcal{B}([0,T])\otimes \mathcal{F}^{\mathbf{X}}_T )-$measurable, and therefore also the function \begin{equation*}
    F^n(t,\mathbf{x}):=F_0^n(\mathbf{x}_0)1_{\{0\}}(t) + \sum_{j=1}^{m_n-1}1_{(t_j^n,t_{j+1}^n]}(t)F_j^n(\mathbf{x}|_{[0,t_j]}).
    \end{equation*} Finally, define the jointly measurable function $F(t,\mathbf{x}):=\limsup_{n\to \infty}F^n(t,\mathbf{x})$, and notice that for almost every $(t,\omega)$, we have 
    \begin{align*}
        F(t,\mathbf{X}(\omega))& =\limsup_{n\to \infty}\bigg(F_0^n\big (\mathbf{X}_0(\omega)\big )1_{\{0\}}(t) + \sum_{j=1}^{m_n-1}1_{(t_j^n,t_{j+1}^n]}(t)F_j^n\big (\mathbf{X}|_{[0,t_j]}(\omega)\big )\bigg) \\ & = \limsup_{n\to \infty}\bigg(\xi^n_0(\omega)1_{\{0\}}(t) + \sum_{j=1}^{m_n-1}1_{(t_j^n,t_{j+1}^n]}(t)\xi^n_{t_{j}}(\omega)\bigg) = A_t(\omega).
    \end{align*} 
    Next, for any element $\mathbf{x}|_{[0,t]}\in \Lambda^{\alpha}_T$, we let $(t,\tilde{\mathbf{x}})\in [0,T]\times \widehat{\mathscr{C}}^{\alpha}_g([0,T],\mathbb{R}^{d+1})$, where $\tilde{\mathbf{x}}$ is the geometric rough-path lift of $[0,T] \ni u \mapsto (u,x_{u \land t})$, see Remark \ref{rmk:polish}. The map $\Xi:\Lambda^{\alpha}_T \longrightarrow [0,T]\times \widehat{\mathscr{C}}^{\alpha}_g$ with $\Xi(\mathbf{x}|_{[0,t]}):=(t,\tilde{\mathbf{x}})$ is continuous and thus especially measurable. Define the composition $f:= F \circ \Xi$, which is a measurable map $f:\Lambda^{\alpha}_T \longrightarrow \mathbb{R}$, such that 
    \begin{equation*}
        f\big(\mathbf{X}|_{[0,t]}(\omega)\big)= A_t(\omega),
    \end{equation*} 
    for almost every $(t,\omega)$, which is exactly what was claimed.
\end{proof} 

\subsection{A Stone--Weierstrass result for robust signatures}\label{sec:SW}

The goal of this section is to present a Stone-Weierstrass type of result for continuous functionals $f:\Lambda_T^{\alpha} \longrightarrow \R$, which will be the key ingredient for the main result in Section \ref{sec:L2approxsection}. To this end, consider the set of linear functionals of the signature 
\begin{equation*}
     L_{\tmop{sig}}(\Lambda^{\alpha}_T) = \{ \Lambda^{\alpha}_T \ni \mathbf{X}|_{[0,t]} \mapsto \langle \Sig{X}_{0,t},\ell \rangle: \ell\in \mathcal{W}^{d+1}  \}\subseteq C(\Lambda_T^{\alpha},\R),
\end{equation*} where we recall that $\mathcal{W}^{d+1}$ denotes the linear span of words, see Section \ref{presectionRP}. Note that continuity of $\mathbf{X}|_{[0,t]} \mapsto \langle \Sig{X}_{0,t},\ell \rangle$ follows from continuity of the \emph{truncated} signature, i.e., $\mathbf{X}|_{[0,t]} \mapsto \mathbf{X}^{\le K}_{0,t}$ -- a consequence of Lyons' extension theorem, see Lyons \cite[Theorem 3.7]{lyons1998differential} -- for $K$ large enough, as any $\ell \in \mathcal{W}^{d+1}$ only has finitely many non-zero coefficients.
A similar set was considered in Kalsi et al. \cite[Definition 3.3]{kalsi2020optimal} with respect to $p$-rough paths, and the authors prove that \textit{restricted} to a compact set $\mathcal{K}$ on the space of time-augemented rough paths, the set $L_{\tmop{sig}}$ is dense in $C (\mathcal{K},
   \mathbb{R})$. In words, restricted to compacts, continuous functionals on the path space $\Lambda_T^{\alpha}$ can be approximated by linear functionals of the signature. However, since such path spaces are not even locally compact, it is desirable to drop the need of a compact set $\mathcal{K}$.
   
   An elegant way to circumvent the requirement of a compact set, is to consider so-called \textit{robust signatures}, introduced in Chevyrev and Oberhauser \cite{chevyrev2018signature}. Loosely speaking, the authors construct a so-called \textit{tensor-normalization} $\lambda$, see \cite[Proposition 14 and Example 4]{chevyrev2018signature}, on the state-space of the signature $T((\R^{d+1}))$, which is a continuous and injective map $$\lambda: T((\R^{d+1}))\longrightarrow \big \{ \mathbf{a}\in T((\R^{d+1})):\Vert \mathbf{a} \Vert \leq R \big \}, \quad R > 0,$$ and they call $\lambda(\Sig{X})$ the robust signature. The continuity of the latter is with respect to the Banach-space norm $\Vert \mathbf{a} \Vert = (\sum_{k\geq 0} \Vert \mathbf{a}^{(k)}\Vert^2_{(\mathbb{R}^{d+1})^{\otimes k}} )^{1/2}$ for  $\mathbf{a}\in T((\mathbb{R}^{d+1}))$ , see \cite[Definition 10]{chevyrev2018signature}. This motivates to define the set \begin{equation*}
     L^{\lambda}_{\tmop{sig}}(\Lambda^{\alpha}_T) =  \big \{ \Lambda^{\alpha}_T \ni \mathbf{X}|_{[0,t]} \mapsto \langle \lambda(\Sig{X}_{0,t}),\ell \rangle: \ell\in \mathcal{W}^{d+1} \big \}\subseteq C_b(\Lambda_T^{\alpha},\R).
\end{equation*} A general version of the Stone-Weierstrass result given in \cite{giles1971generalization}, leads to the following result, which was stated already in Giles \cite[Theorem 26]{chevyrev2018signature}, and we present the proof here for completeness. 
\begin{lemma}
  \label{density} Let $\alpha \in (0,1) $. Then the set $L_{\tmop{sig}}^{\lambda}(\Lambda^{\alpha}_T)$ is dense in $C_b
  (\Lambda^{\alpha}_T, \mathbb{R})$ with respect to the strict topology. More precisely, for any $f \in C_b
  (\Lambda^{\alpha}_T, \mathbb{R})$ we can find a sequence $(f_n)_{n\in \N} \subseteq L_{\tmop{sig}}^{\lambda}(\Lambda^{\alpha}_T)$, such that 
  $$
  \Vert f-f_n \Vert_{\infty,\psi}:= \sup_{\mathbf{x}\in \Lambda_T^{\alpha}} \big |\psi(\mathbf{x})(f(\mathbf{x})-f_n(\mathbf{x}))\big | \rightarrow 0, \text{ as } n\to \infty, \quad \forall \psi \in B_0(\Lambda_T^{\alpha}),
  $$ 
  where $B_0(\Lambda_T^{\alpha})$ denotes the set of functions $\psi:\Lambda_T^{\alpha} \rightarrow \mathbb{R}$, such that for all $\epsilon > 0$ there exists a compact set $K\subseteq \Lambda_T^{\alpha}$ with $\sup_{\mathbf{x}\in \Lambda_T^{\alpha}\setminus{K}}|\psi(\mathbf{x})|< \epsilon$.
\end{lemma} 

\begin{proof} This result is a consequence of the general Stone-Weierstrass result proved in \cite{giles1971generalization}, see also \cite[Theorem 9]{chevyrev2018signature}. From the latter, we only need to check the inclusion $L_{\text{sig}}^{\lambda} \subseteq C_b(\Lambda_T^{\alpha},\mathbb{R})$ is a subalgebra, such that 
    \begin{enumerate}
    \item $L_{\text{sig}}^{\lambda}$ separates points, that is $\forall x \neq y$ there
    exists $f \in L_{\text{sig}}^{\lambda}$ such that $f (x) \neq f (y)$.
    \item $L_{\text{sig}}^{\lambda}$ contains non-vanishing functions, that is $\forall
    x$ there exists $f \in L_{\text{sig}}^{\lambda}$ such that $f (x) \neq 0$.
\end{enumerate}
  As mentioned in the beginning of Section \ref{sec:SW}, $\mathbf{X}|_{[0,t]}\mapsto \langle \mathbf{X}^{<\infty}_{0,t},\ell\rangle$ is continuous for any $\ell\in \mathcal{W}^{d+1}$, and since the tensor-normalization is continuous and bounded, it follows that $\mathbf{X}|_{[0,t]}\mapsto \langle \lambda(\mathbf{X}^{<\infty}_{0,t}),\ell\rangle$ is continous and bounded. To see that $L_{\tmop{sig}}^{\lambda} \subseteq C_b(\Lambda_T^{\alpha},\mathbb{R})$ is a subalgebra, we fix $\phi_1, \phi_2
  \in L_{\tmop{sig}}^{\lambda}$. By definition, there exists a linear combination of words $\ell_1, \ell_2 \in \mathcal{W}^{d+1}$ such that $\phi_i (\mathbf{X} |_{[0, s]}
  \nobracket) = \langle \lambda  (\mathbf{X}^{< \infty}_{0,s}),\ell_i
  \rangle$. We clearly have 
  \begin{equation*}
      \phi (\mathbf{X} |_{[0, s]} \nobracket) \assign \phi_1
     (\mathbf{X} |_{[0, s]} \nobracket) + \phi_2
     (\mathbf{X} |_{[0, s]} \nobracket) = \big \langle 
     \lambda  (\mathbf{X}^{< \infty}_{0,s}),\ell_1 + \ell_2 \big \rangle \in
     L_{\tmop{sig}}^{\lambda}. 
  \end{equation*} 
  Now assume $\ell_1,\ell_2$ are words $\ell_1=w$ and $\ell_2= v$. By definition of tensor-normalization \cite[Definition 12]{chevyrev2018signature}, for some positive function $\Psi: T((\R^{d+1})) \to ]0,+\infty[$, we have \begin{align*}
     \phi_1
     (\mathbf{X} |_{[0, s]} \nobracket) \cdot \phi_2
     (\mathbf{X} |_{[0, s]} \nobracket) & = \big \langle  \lambda  (\mathbf{X}^{< \infty}_{0,s}), w\big \rangle \big \langle  \lambda  (\mathbf{X}^{< \infty}_{0,s}), v \big \rangle \\ & = \Psi  (\mathbf{X}^{< \infty}_{0,s})^{|w|+|v|} \langle \mathbf{X}^{< \infty}_{0,s}, w\rangle\langle  \mathbf{X}^{< \infty}_{0,s}, v\rangle \\ & = \Psi  (\mathbf{X}^{< \infty}_{0,s})^{|w|+|v|}  \langle \mathbf{X}^{< \infty}_{0,s}, w \shuffle v\rangle,
  \end{align*} 
  where we used that $\mathbf{X}^{<\infty}\in G(\R^d)$ for the last equality. But by definition of the shuffle product \eqref{def:shuffle-product}, it follows that $w\shuffle v = \sum_j u_j$, where $u_j$ are words with $|u_j| = |w|+|v|$, and hence \begin{align*}
      \phi_1
     (\mathbf{X} |_{[0, s]} \nobracket) \cdot \phi_2
     (\mathbf{X} |_{[0, s]} \nobracket) & = \sum_j\Psi  (\mathbf{X}^{< \infty}_{0,s})^{|w|+|v|} \langle \mathbf{X}^{< \infty}_{0,s}, u_j\rangle = \big \langle  \lambda  (\mathbf{X}^{< \infty}_{0,s}), w \shuffle v \big \rangle,
  \end{align*} so that the product lies in $L_{\tmop{sig}}^{\lambda}$. The same reasoning can be extended to linear combination of words $\ell_1,\ell_2$, and thus the set $L_{\mathrm{sig}}^{\lambda}$ is indeed a subalgebra in $C_b(\Lambda^{\alpha}_T,\mathbb{R})$. Now let
  $\mathbf{X}, \mathbf{Y} \in \Lambda_T^{\alpha}$, such that
  $\mathbf{X} \neq \mathbf{Y}$. As remarked in Section \ref{presectionRP}, since we are working with rough-path lifts of time-augmented paths $(t,X_t)$, the signature map is injective. Moreover, by definition \cite[Definition 12]{chevyrev2018signature}, the map $\lambda$ is also injective, therefore $\lambda
  (\mathbf{X} ^{< \infty}) \neq \lambda (\mathbf{Y}
  ^{< \infty})$. This in particular implies that there exists $\ell\in \mathcal{W}^{d+1}$ such that $\langle\lambda
  (\mathbf{X} ^{< \infty}),\ell \rangle \neq \langle\lambda
  (\mathbf{Y} ^{< \infty}),\ell \rangle $. Defining $f \in L_{\tmop{sig}}^{\lambda}$ by $f(\mathbf{x})=\langle \lambda
  (\mathbf{x} ^{< \infty}),\ell \rangle$, it follows that $f(\mathbf{X}) \neq f(\mathbf{Y})$ and thus
  $L_{\tmop{sig}}^{\lambda}$ separates points. Finally, since $1 =
  \langle  \lambda  (\mathbf{X}^{< \infty} ), \emptyset \rangle \in
  L_{\tmop{sig}}^{\lambda}$, the claim follows.
\end{proof} 

\subsection{Approximation with robust signatures}\label{sec:L2approxsection}

We are now ready to state and proof the main result of this section. We will assume the following \begin{Assumption}\label{ass:measure}\begin{itemize}
        \item[(i)] $\alpha \in (0,1)$ is such that $q := \frac{1}{\alpha}\notin \mathbb{N}$.
        \item[(ii)] $\mu$ is a measure on the Borel space $(\widehat{\mathscr{C}}^{\alpha}_g,\mathcal{B}(\widehat{\mathscr{C}}^{\alpha}_g))$ such that \begin{equation*} 
    \mu ( \widehat{\mathscr{C}}^{\alpha}_g )< \infty\text{ and }  \mu  ( \widehat{\mathscr{C}}^{\alpha}_g\setminus \widehat{\mathscr{C}}^{\beta}_g )=0, \quad \forall \alpha < \beta  < \frac{1}{\lfloor q \rfloor }.\end{equation*}
    \end{itemize}

\end{Assumption} Before stating the main result of this section, we give an example to clarify the role of Assumption \ref{ass:measure}.
\begin{example}
    Let $X$ be a $d$-dimensional Brownian motion and fix $\alpha \in (1/3,1/2)$, that is (i) in Assumption \ref{ass:measure} holds with $2<q<3$. Denote by $\mathbf{X}$ the geometric $\alpha-$Hölder rough-path lift of the time-augmentation $(t,X_t)$, that is, define the second level using Stratonovich integration, see Friz and Hairer \cite[Proposition 3.5]{friz2020course}. As discussed in Remark \ref{rmk:geometricRV}, we can see $\mathbf{X}$ as $\widehat{\mathscr{C}}_{g}^{\alpha}$-valued random variable with law $\mu_{\mathbf{X}}$. Applying a Kolmogorov type of result for rough paths, see for instance \cite[Theorem 3.1]{friz2020course}, it follows that $\mu_{\mathbf{X}}$ assigns full measure to $\widehat{\mathscr{C}}_{g}^{\beta}$ for all $\beta \in (\alpha,1/2)$, and thus Assumption \ref{ass:measure} (ii) holds. Similarly we can treat geometric lifts of more general processes such as semimartingales, but also fractional Brownian motion by replacing $1/2$ by the Hurst parameter $H$.
\end{example}
 The following theorem shows that under Assumption \ref{ass:measure}, we can approximate any functional in $L^p(\Lambda_T^{\alpha},\widehat{\mu})$, where we recall from Remark \ref{rmk:geometricRV} that  $\widehat{\mu}$ is the pushforward measure of $dt \otimes d\mu$ on $\Lambda_T^{\alpha}$, by linear functionals of the robust signature with respect to the $L^p$-norm.

\begin{theorem}\label{mainresultappro}
    Let $\alpha \in (0,1)$ and $\mu$ be a measure on $(\widehat{\mathscr{C}}^{\alpha}_g,\mathcal{B}(\widehat{\mathscr{C}}^{\alpha}_g))$ such that Assumption \ref{ass:measure} holds true. %\textcolor{blue}{Moreover, suppose $\widehat{\mu}$ is the pushforward measure of $dt \otimes d\mu$ defined in Remark \ref{rmk:geometricRV}}.
    Then for all $f\in L^{p}(\Lambda_T^{\alpha},\widehat{\mu})$, $1 \leq p < \infty$, there exists a sequence $(f_n)_{n\in \mathbb{N}}\subset L_{\text{sig}}^{\lambda}(\Lambda_T^{\alpha})$ such that $\Vert f-f_n \Vert_{L^p}\rightarrow 0$ as $n\to \infty$.
\end{theorem} Before proving this result, let us show the following immediate consequence for random geometric rough paths, which will be of particular importance in Section \ref{sec: OptimalStoppingSection}.
\begin{corollary}\label{mainresultcoro}
Let $\geoRP{X}$ be a stochastic process on $\widehat{\mathscr{C}}^{\alpha}_g$ such that Assumption \ref{ass:measure} holds for its law $\mu_{\mathbf{X}}$. Then for all $A\in \mathbb{H}^2$, see \eqref{def:prognorm}, there exists a sequence of linear functionals $(f_n)_{n\in \mathbb{N}}\subset L_{\text{sig}}^{\lambda}(\Lambda_T^{\alpha})$ such that $\Vert A-f_n(\mathbf{X}|_{[0,\cdot]}) \Vert_{\mathbb{H}^2} \rightarrow 0$ as $n\to \infty$.
\end{corollary}

\begin{proof} By Lemma \ref{progressivLemma}, there exists a measurable function $f:\Lambda_T^{\alpha} \rightarrow \mathbb{R}$ such that $A_t=f(\mathbf{X}|_{[0,t]})$. Applying a standard change of measure result for the push-forward measure, see for example Bogachev and Ruas \cite[Theorem 3.6.1]{bogachev2007measure}, and denoting by $\phi$ the quotient map given in Definition \ref{def:stoppedRP}, we have \begin{align*}
    \Vert A_t-f_n(\mathbf{X}|_{[0,\cdot]}) \Vert_{\mathbb{H}^2}^2 &= E\bigg [\int_0^T\big (f(\mathbf{X}|_{[0,t]})-f_n(\mathbf{X}|_{[0,t]})\big )^2dt  \bigg ]  \\ & = \int_{\widehat{\mathscr{C}}_g^{\alpha}}\int_0^T(f\circ \phi-f_n \circ \phi)(t,\mathbf{X})^2dtd\mu_{\mathbf{X}} \\ & = \int_{\Lambda_T^{\alpha}}(f-f_n)^2d\widehat{\mu}= \Vert f-f_n \Vert_{L^2}^2\rightarrow 0,
\end{align*} as $n \to \infty$, where the convergence follows from Theorem \ref{mainresultappro}.
\end{proof}
The proof of Theorem \ref{mainresultappro} will make use of two lemmas. The first one is very elementary, and in the language of probability theory, it states that for every random variable $\xi$ in $\mathbb{R}_+$, we can find a strictly increasing and integrable function $\eta$, that is $E[\eta(\xi)]<\infty$. \begin{lemma}\label{integrabilitylemma}
    Let $(E,\mathcal{E},\mu)$ be a finite measure space, and consider a measurable function $\xi:(E,\mathcal{E}) \rightarrow (\mathbb{R}_+,\mathcal{B}(\mathbb{R}_+))$. Then there exists a strictly increasing function $\eta:\mathbb{R}_{+}\rightarrow \mathbb{R}_{+}$, such that $\eta(x) \rightarrow \infty$ as $x\to \infty$ and $\int_E (\eta \circ \xi) d\mu<\infty$.
\end{lemma} \begin{proof}
    Let $\nu$ be the push-foward of $\mu$ under $\xi$, that is $\nu(A):=\mu(\xi^{-1}(A))$ for all sets $A\in \mathcal{B}(\mathbb{R}_+)$. Then for any $\epsilon>0$ we can find $R>0$ large enough, such that $\nu(]R,+\infty[)\leq \epsilon$. In particular, for any strictly decreasing  sequence $(a_n)_{n\geq 0}$, such that $a_n \searrow 0$, we can find a strictly increasing sequence $(R_n)_{n\in \mathbb{N}}$ with $R_1 > 0$, such that $\nu(]R_n,\infty[) \leq \frac{a_n}{n^2}$. Now we can define a strictly increasing function $\eta$ as follows: Let $\eta(0)=0$ and for all $n\in \mathbb{N}$ define $\eta(R_n)= \frac{1}{a_{n-1}}$, and linearly interpolate on the intervals $[R_n,R_{n+1}[$. Then, setting $R_0=0$, and using a change of measure Bogachev and Ruas \cite[Theorem 3.6.1]{bogachev2007measure}, we have \begin{align*}
        \int_E(\eta \circ \xi)d\mu = \int_0^\infty \eta d\nu \leq \sum_{n\geq 0} \frac{1}{a_n}\nu(]R_n,+\infty[)\leq \frac{\nu([0,\infty[)}{a_0} + \sum_{n\geq 1}\frac{1}{n^2}<\infty.
    \end{align*}
\end{proof} The next lemma will be the key ingredient to apply the Stone-Weierstrass result Lemma \ref{density} in the main result. \begin{lemma}\label{vanishinglemma} Let $\alpha<\beta$ and define $\bar{\psi}:\Lambda_T^{\alpha} \rightarrow \mathbb{R}$ by $\bar{\psi}(\mathbf{x}) := 1_{\Lambda_T^{\beta}}(\mathbf{x}) ( \frac{1}{1+\eta(\vertiii{\tilde{\mathbf{x}}}_{\beta})} )^{1/p}$, where $\eta:\mathbb{R}_+ \rightarrow \mathbb{R}_+$ is strictly increasing such that $\eta(x) \rightarrow \infty$ as $x\to \infty$. Then we have $\bar{\psi} \in B_0(\Lambda_T^{\alpha})$, that is for all $\epsilon > 0 $ there exists $K\subseteq \Lambda_T^{\alpha}$ compact, such that $\sup_{\mathbf{x}\in K^c}\bar{\psi}(\mathbf{x})\leq \epsilon$.
\end{lemma} 
\begin{proof} Recall that an element $\mathbf{x} \in \Lambda_T^{\alpha}$ can be written as $\mathbf{x}= \mathbf{x}|_{[0,t]} \in \widehat{\mathscr{C}}^{\alpha}_g([0,t])$, where $\mathbf{x}|_{[0,t]}$ is the geometric rough-path lift of some time-augmented, $\alpha$-Hölder continuous $[0,t] \ni u \mapsto (u,\omega_u)$. Moreover, recall that we define $\tilde{\mathbf{x}}\in \widehat{\mathscr{C}}^{\alpha}_g([0,T])$ to be the geometric rough-path lift of $u \mapsto (u,x_{t\land u}) $. If we can show that for any $R>0$, the sets $$B_R = \{\mathbf{x} \in \Lambda_T^{\beta}: \vertiii{\tilde{\mathbf{x}}}_{\beta}\leq R\}\subseteq \Lambda_T^{\alpha}$$ are compact, then we are done. Indeed, in this case we have that for any $\epsilon \in (0,1)$, we can choose $\widehat{R}\geq \eta^{-1}(\frac{1-\epsilon^p}{\epsilon^p})$ and then $$\psi(\mathbf{x}) = \bigg (\frac{1}{1+\eta(\vertiii{\tilde{\mathbf{x}}}_{\beta})}\bigg ) ^{1/p} \leq \frac{1}{\sqrt{1+\widehat{R}}} \leq \epsilon, \quad \forall \mathbf{x} \in \Lambda_T^{\alpha} \setminus B_{\widehat{R}},$$ and therefore $\bar{\psi} \in B_0(\Lambda_T^{\alpha})$. Now to prove compactness, we can notice that by definition of the quotient map $\phi$, see Definition \ref{def:stoppedRP}, we have $$B_R\subseteq \phi \Big ( [0,T] \times \big \{ \mathbf{x}\in \widehat{\mathscr{C}}^{\beta}_g([0,T],\mathbb{R}^{d+1}): \vertiii{\mathbf{x}}_{\beta} \leq R\big \}\Big ),$$ since for all $\mathbf{x}=\mathbf{x}|_{[0,t]}\in B_R$ we have $\mathbf{x}= \phi(t,\tilde{\mathbf{x}})$ by construction. Since $\phi$ is continuous, it is enough the show that $[0,T] \times \{\mathbf{x} \in \widehat{\mathscr{C}}^{\beta}_g:\vertiii{\mathbf{x}}_{\beta}\leq R \}$ is compact in $[0,T] \times \widehat{\mathscr{C}}^{\alpha}_g$, which by Tychonoffs theorem is true if the sets $\{\mathbf{x} \in \widehat{\mathscr{C}}^{\beta}_g:\vertiii{\mathbf{x}}_{\beta}\leq R \}$ are compact in $\widehat{\mathscr{C}}^{\alpha}_g$. But the latter follows from the general fact that $\beta$-Hölder spaces are compactly embedded in $\alpha$-Hölder spaces for $\alpha<\beta$. This can be proved by applying the Arzelà–Ascoli theorem together with an interpolation argument for the equicontinuous and $\vertiii{\cdot}_{\beta}$-bounded subsets of $\widehat{\mathscr{C}}^{\alpha}_g$, which was carried out in Cuchiero et al. \cite[Theorem A.3]{cuchiero2023global} for example. Thus we can conclude that $B_R\subseteq \Lambda_T^{\alpha}$ is compact, which finishes the proof.
\end{proof}
Finally, we are ready to proof the main result.

\begin{proof}{\textbf{of Theorem \ref{mainresultappro}}} First, recall that the measure $\widehat{\mu}$ is defined as the pushforward of the product measure $dt \otimes d\mu$ via the surjective map $\phi$ from Definition 2.1 introduced in in Remark 2.3. We can easily see that for $\alpha,\beta$ as in Assumption 2.6, we have \begin{equation*}\phi^{-1}(\Lambda_T^{\alpha}\setminus \Lambda_T^{\beta}) =\phi^{-1}(\Lambda_T^{\alpha})\setminus \phi^{-1}(\Lambda_T^{\beta}) \subseteq [0,T] \times \widehat{\mathscr{C}}^{\alpha}_g \setminus \widehat{\mathscr{C}}^{\beta}_g,\end{equation*} and thus $\widehat{\mu}(\Lambda_T^{\alpha}\setminus \Lambda_T^{\beta}) \leq dt \otimes d\mu ( [0,T] \times \widehat{\mathscr{C}}^{\alpha}_g \setminus \widehat{\mathscr{C}}^{\beta}_g)=0$. Fix $\epsilon>0$. For any value $K > 0$, we can define the function $f_K(x) := 1_{\{|f(x)| \leq K\}}(x)f(x)$, and notice that we have $\Vert f-f_K\Vert_{L^p} \rightarrow 0$ as $K\to \infty$ by dominated convergence. Hence we can find a $K_{\epsilon}>0$ such that $\Vert f-f_{K_{\epsilon}}\Vert_{L^p} \leq \epsilon /3$. Since $\widehat{\mu}$ is a finite measure on $\Lambda^{\alpha}_T$, by Lusin's theorem, we can find a closed set $C_{\epsilon} \subset \Lambda^{\alpha}_T$, such that $f_{K_{\epsilon}}$ restricted to $C_{\epsilon}$ is continuous, and $\widehat{\mu}(\Lambda^{\alpha}_T \setminus C_{\epsilon}) \leq \epsilon^p/(6K_{\epsilon})^p$. By Tietze's extension theorem, we can find a continuous extension $\widehat{f}_{\epsilon}\in C_b(\Lambda^{\alpha}_T,[-K_{\epsilon},K_{\epsilon}])$ of $f_{K_{\epsilon}}$ such that\begin{equation*}
    \Vert f_{K_{\epsilon}}-\widehat{f}_{\epsilon} \Vert_{L^p}^p =  \int_{\Lambda^{\alpha}_T \setminus C_{\epsilon}}|f_{\epsilon}-f_{K_{\epsilon}}|^pd\widehat{\mu} \leq (2K_{\epsilon})^p\widehat{\mu}({\Lambda^{\alpha}_T \setminus C_{\epsilon}}) = (\epsilon/3)^p.
\end{equation*} We are left with approximating $\widehat{f}_{\epsilon} \in C_b(\Lambda^{\alpha}_T,\mathbb{R})$ by linear functionals of the robust signature, that is applying Lemma \ref{density}. By Assumption \ref{ass:measure} we can choose $\beta \in (\alpha,\frac{1}{\lfloor q \rfloor})$ and from Lemma \ref{integrabilitylemma} we know that there exists an increasing function $\eta:\mathbb{R}_+ \to \mathbb{R}_+$, such that $\eta(x) \rightarrow \infty$ as $x\to \infty$ and $\int_{\Lambda_T^{\alpha}}\eta(\vertiii{\tilde{\mathbf{x}}}_{\beta})d\widehat{\mu}(x) <\infty$, where $\tilde{\mathbf{x}}$ is the extension of the stopped rough path from the interval $[0,t]$ to $[0,T]$, see also Remark \ref{rmk:polish}. Define the function $\psi:\Lambda_T^{\beta} \rightarrow \mathbb{R}_+$ by $\psi(\mathbf{x}|_{[0,t]}):=  (\frac{1}{1+\eta(\vertiii{\tilde{\mathbf{x}}}_{\beta})} )^{1/p}$.  In Lemma \ref{vanishinglemma} we saw that $\bar{\psi}(\mathbf{x}) := 1_{\Lambda_T^{\beta}}(\mathbf{x})\psi(\mathbf{x})$ belongs to $B_0(\Lambda_T^{\alpha})$, that is for all $\delta >0$ there exists a compact set $K\subseteq \Lambda_T^{\alpha}$, such that $\sup_{x\in K^c}\bar{\psi}(x) \leq \delta$. Notice that Lemma \ref{integrabilitylemma} yields an increasing function $\eta$, which is used to build the $\bar{\psi}$ in Lemma \ref{vanishinglemma}, in a way that $\Xi:= \int_{\Lambda_T^{\beta}} \frac{1}{\bar{\psi}^p}d\widehat{\mu}$ is finite, which will be needed below. If $\eta$ is an arbitrary increasing function, the integrability of $\frac{1}{\bar{\psi}^p}$ needs to be assumed (e.g. for $\eta(x)=x$ we would need finite first moment of the rough-path norm). By Lemma \ref{density}, we can find $f_{\epsilon} \in L^{\lambda}_{\text{sig}}(\Lambda_T^{\alpha})$, such that \begin{equation*}
    \Vert \widehat{f}_{\epsilon}-f_{\epsilon} \Vert_{\infty,\bar{\psi}}^p \leq \epsilon^p/(3^p\Xi).
 \end{equation*}  Using that $\mu$ assigns full measure to the subspace $\Lambda_T^{\beta} \subseteq \Lambda_T^{\alpha}$, we have \begin{align*}
    \Vert \widehat{f}_\epsilon -f_{\epsilon} \Vert_{L^p}^p  = \int_{\Lambda_T^{\beta}}|\widehat{f}_{\epsilon}-f_{\epsilon}|^pd{\widehat{\mu}} & \leq \sup_{x\in \Lambda_T^{\beta}}\Big ( \psi(x)\big(\widehat{f}_{\epsilon}(x)-f_{\epsilon}(x)\big ) \Big)^p \int_{\Lambda_T^{\beta}}\frac{1}{\psi^p}d{\widehat{\mu}} \\ & \leq \Vert \widehat{f}_{\epsilon}-f_{\epsilon} \Vert_{\infty,\bar{\psi}}^p \Xi \leq (\epsilon/3)^p.
\end{align*} Finally, we can conclude by the triangle inequality \begin{equation*}
    \Vert f-f_{\epsilon}\Vert_{L^p} \leq \Vert f-f_{K_{\epsilon}}\Vert_{L^p} + \Vert f_{K_{\epsilon}}-\widehat{f}_{\epsilon}  \Vert_{L^p} + \Vert \widehat{f}_{\epsilon} -f_\epsilon \Vert_{L^p} \leq \epsilon.
\end{equation*}
\end{proof} 

%\begin{remark} Similar global approximation results with (not necessarily robust) signatures were studied in \cite{cuchiero2023global}, for different types of path spaces. They introduce so called weighted spaces $(\mathcal{X},\psi)$, where $\mathcal{X}$ is a topological space and the weight-function $\psi$ is such that $\psi^{-1}(]0,R])$ is compact in $\mathcal{X}$, and they prove a Stone-Weierstrass result for weighted spaces, see \cite[Section 3]{cuchiero2023global}. Put in our context, the weight-function would be $\Psi$, and $\Psi^{-1}(]0,R])$ is compact by compact-embedding, see Lemma \ref{vanishinglemma}. Thus the pair $(\Lambda_T^{\alpha},\Psi)$ is a weighted-space, and the approximation result can be compared with \cite[Theorem 5.4]{cuchiero2023global} for robust-signatures. \end{remark}

\section{Optimal stopping with signatures}\label{sec: OptimalStoppingSection} In this section we exploit the signature approximation theory presented in Section \ref{sec:L2approxsection}, in order solve the optimal stopping problem in a general setting. \subsection{Framework and problem formulation} \label{sec:framework}
Suppose we have a complete, filtered probability space $(\Omega,\mathcal{F},\mathbb{F}=(\mathcal{F}_t)_{t\in [0,T]},P)$ for some $T>0$, fulfilling the usual conditions, and fix $\alpha \in (0,1)$. For any $\mathbb{F}-$adapted and $\alpha$-Hölder continuous stochastic process $(X_t)_{t\in[0,T]}$, taking values in $\R^d$ with $X_0=x_0$, we consider

\begin{itemize} \item $\geoRP{X} \in \widehat{\mathscr{C}}^{\alpha}_g$ a geometric $\alpha$-Hölder rough-path lift of $(t,X_t)$, such that its law $\mu_{\mathbf{X}}$ fulfills Assumption \ref{ass:measure}, \item $\Sig{X}$ the robust %\footnote{Notice the small abuse of notation here, as the robust signature is given by $\lambda(\geoRP{X}^{<\infty})$ for some tensor-normalization $\lambda$. For the rest of this paper, we fix such an $\lambda$ and write $\geoRP{X}^{<\infty}$ for $\lambda(\geoRP{X}^{<\infty})$.}
rough-path signature introduced in Section \ref{sec:L2approxsection}, \item $(Z_t)_{t\in [0,T]}$ is a real-valued, $\mathbb{F}^{\mathbf{X}}-$adapted process such that $\sup_{t\in [0,T]}|Z_t| \in L^2$. \end{itemize} The optimal stopping problem then reads \begin{equation}\label{eq:OSPmain}
    y_0 = \sup_{\tau \in \mathcal{S}_0}E [Z_{\tau}],
\end{equation} where $\mathcal{S}_0$ denotes the set of $\mathbb{F}^{\mathbf{X}}-$stopping-times on $[0,T]$.
%\begin{Assumption}\label{ass:standingass} For all $\alpha < \gamma$ we have $\vertiii{\mathbf{X}}_{(\alpha,N)}\in L^1$ and $\sup_{t\in [0,T]}|Z_t| \in L^2$. \end{Assumption} 
\begin{remark} Notice that the framework described above is very general in two ways: First, we only assume $\alpha$-Hölder continuity for the state process $X$, including in particular non-Markovian and non-semimartingale regimes, which one for instance encounters in rough volatility models, see Section~\ref{sec:AmericanOptions}. Secondly, considering a projection onto the first coordinate $\geoRP{X} \mapsto (t,X_t)$, for any payoff function $\phi :[0,T] \times \R^d \rightarrow \mathbb{R}$ our framework includes the more common form of the optimal stopping problem $$ y_0 = \sup_{\tau \in \mathcal{S}_0}E \big [\phi(\tau,X_{\tau})\big ].$$ \end{remark}

\begin{remark}
In general, there is no canonical way of lifting a process $(t,X_t)$ to a random rough path $\mathbf{X}$, and often careful justification is required, for instance based on a rough-path version of the Kolmogorov criterion, see Friz and Hairer \cite[Theorem 3.1]{friz2020course}. However, for a big class of process (e.g. semimartingales, Gaussian processes, one-dimensional processes) there are canonical choices for random geometric rough-path lifts, and we will explain in Section \ref{sec:Numerics} how to do it. Moreover, we will always look at lifts such that $\mathbb{F}^{\mathbf{X}}=\mathbb{F}^{X}$, that is the optimal stopping problem has the same underlying information when observing $\mathbf{X}$ or $X$.
\end{remark}

\subsection{Primal optimal stopping with signatures}\label{sec: LSsection} 
First we present a method to compute a lower-biased approximation $y_0^{L}\leq y_0$ to the optimal stopping problem \eqref{eq:OSPmain}. More precisely, we construct a regression-based approach, generalizing the famous algorithm from Longstaff and Schwartz~\cite{longstaff2001valuing}, returning a sub-optimal exercise strategy. Let us first quickly describe the main idea of most regression-based approaches.

Replacing the interval $[0,T]$ by a finite grid $\{0=t_0<t_1< \dots <t_N = T\}$, the discrete optimal stopping problem reads 
\[
y_{0}^N =\sup_{\tau \in \mathcal{S}_0^N}E [Z_{\tau}],
\]
     where $\mathcal{S}_n^N$ is the set of stopping times taking values in $\{t_n,\dots,t_N\}$ for $n=0,\dots,N$, with respect to the discrete filtration $\mathbb{F}^{\geoRP{X},N}=(\mathcal{F}^{\geoRP{X}}_{t_n})_{n=0,\dots,N}$. We define the discrete \textit{Snell-envelope} by \begin{equation}\label{eq:SnellEnvelope}
    Y^{N}_{t_n} = \esssup_{\tau \in \mathcal{S}_n^N}E [Z_{\tau}|\mathcal{F}^{\geoRP{X}}_{t_n} ], \quad 0 \leq n \leq N,
\end{equation} and one can show that $Y^N$ satisfies the discrete dynamic programming principle (DPP) 
    \begin{equation}
    Y^{N}_{t_n}= \max \big(Z_{t_n},E [Y^{N}_{t_{n+1}}|\mathcal{F}^{\geoRP{X}}_{t_n} ] \big ), \quad n=0,\dots,N-1,\label{eq:DPP}
\end{equation} see for instance Peskir and Shiryaev \cite[Theorem 1.2]{peskir2006optimal}. Now the key idea of most regression-based approaches, such as for instance Longstaff and Schwartz \cite{longstaff2001valuing}, is that assuming $X$ is a Markov process, one can choose a suitable family of basis functions $(b^{k})$ and apply least-square regression to approximate \begin{equation}\label{eq:approxMarkov}
E [Y^{N}_{t_{n+1}}|\mathcal{F}^{\mathbf{X}}_{t_n} ] \approx \sum_{k = 0}^D\alpha_kb^{k}_n(X_{t_n}), \quad 0 \leq n \leq N-1, \quad \alpha_k \in \mathbb{R}, \quad \forall k \leq D,
\end{equation} and then make use of the DPP to recursively approximate $Y_0^N=y_0^N$. Of course, the approximation of the conditional expectations in \eqref{eq:approxMarkov} heavily relies on the Markov-property, and thus one cannot expect such an approximation to converge in non-Markovian settings. 

Returning to our framework, we need to replace \eqref{eq:approxMarkov} by a suitable approximation for the conditional expectations \begin{equation*}
E[Y^{N}_{t_{n+1}}|\mathcal{F}^{\geoRP{X}}_{t_n}]=f_n(\geoRP{X}|_{[0,t_n]}), \quad 0 \leq n \leq N-1.
\end{equation*} The universality result Theorem~\ref{mainresultappro} now suggests to approximate $f_n$ by a sequence of linear functionals of the robust signature, that is \begin{equation*}
f_n(\geoRP{X}|_{[0,t_n]}) \approx \langle \geoRP{X}^{<\infty}_{0,t_n},\ell\rangle, \quad \ell \in \mathcal{W}^{d+1},
\end{equation*}  where $\mathcal{W}^{d+1}$ is the linear span of words introduced in Section \ref{presectionRP}.

\subsubsection{Longstaff-Schwartz with signatures} In this section we present a version of the Longstaff-Schwartz (LS) algorithm \cite{longstaff2001valuing}, using signature-based least-square regression. A convergence analysis for the LS-algorithm was presented in Clément et al. \cite{clement2002analysis}, and combining their techniques with the universality of the signature, allows us to recover a convergent algorithm.

The main idea of the LS-algorithm is to re-formulate the DPP \eqref{eq:DPP} for stopping times, taking advantage of the fact that optimal stopping times can be expressed in terms of the Snell-envelope. More precisely, it is proved in Peskir and Shiryaev \cite[Theorem 1.2]{peskir2006optimal} that the stopping times $\tau_n := \min\{t_m \geq t_n: Y^{N}_{t_m}=Z_{t_m}, m=n,\dots,N\}$ are optimal in \eqref{eq:SnellEnvelope}, and hence one recursively defines \begin{equation*}
        \begin{aligned}
            \tau_N & = t_N \\ \tau_n &= t_n1_{\big\{Z_{t_n} \geq E [Z_{\tau_{n+1}}|\mathcal{F}^{\geoRP{X}}_{t_n}]\big \}}+\tau_{n+1}1_{\big\{Z_{t_n} < E [Z_{\tau_{n+1}}|\mathcal{F}^{\geoRP{X}}_{t_n} ]\big \}}, \quad n=0,\dots,N-1.
        \end{aligned}
\end{equation*} Now for any truncation level $K\in \mathbb{N}$ for the signature and for some fixed $n=0,\dots,N-1$, we assume that we are given an approximation $\tau^K_{n+1}$ of $\tau_{n+1}$. Then, we approximate the conditional expectation $E[Z_{\tau_{n+1}^K}|\mathcal{F}_{t_n}^{\mathbf{X}}]$ by solving the following minimization problem \begin{equation}\label{eq:orthogonalprojections}
         \ell^*:=\ell^{*,n,K} = \underset{\ell \in \mathcal{W}_{\leq K}^{d+1}}{\operatorname{argmin}}  \Vert Z_{\tau^K_{n+1}}-\langle \mathbf{X}_{0,t_n}^{\leq K},\ell \rangle  \Vert_{L^2}, \quad n=0,\dots, N-1. 
    \end{equation} Setting $\psi^{n,K}(\mathbf{x}) = \langle \mathbf{x}^{\leq K},\ell^{*,n,K}\rangle \in L_{\text{sig}}^{\lambda}$, we can define the following approximating sequence of stopping times \begin{equation*}
    \begin{aligned}
        \tau^K_N & = t_N \\ \tau^K_n &= t_n1_{\big\{Z_{t_n} \geq \psi^{n,K}(\geoRP{X}|_{[0,t_n]})\big \}}+\tau_{n+1}^{K}1_{\big\{Z_{t_n} < \psi^{n,K}(\geoRP{X}|_{[0,t_n]})\big \}}, \quad n=0,\dots,N-1.
    \end{aligned}
\end{equation*} The following result shows convergence as the depth of the signature goes to infinity, and the proof is discussed in Appendix \ref{appendixLS}. \begin{proposition}\label{THM:LSconvergence1} For all $n=0,\dots,N$ we have $$\lim_{K\to \infty }E[Z_{\tau^K_{n}}|\mathcal{F}^{\geoRP{X}}_{t_n}] = E[Z_{\tau_{n}}|\mathcal{F}^{\geoRP{X}}_{t_n}] \text{ in }L^2.$$
 In particular, we have $y_0^{K,N} = \max(Z_{t_0},E[Z_{\tau_1^K}])  \rightarrow y_0^N$ as $K\to \infty$.
    \end{proposition} Let us now describe how to numerically solve \eqref{eq:orthogonalprojections} using Monte-Carlo simulations. Besides the truncation of the signature at some level $K$, we introduce two further approximations steps: First, we replace the signature $\Sig{X}$ by some discretized version $\Sig{X}(J)$, for example piecewise linear approximation of the iterated integrals, on some fine grid $s_0=0<s_1<\dots <s_J = T$, such that $\langle \Sig{X}_{0,t}(J),v \rangle  \rightarrow \langle \Sig{X}_{0,t},v \rangle$ as $J \to \infty$ in $L^2$ for all words $v$ and $t \in [0,T]$. Secondly, % we can explicitely write the minimizer for the linear least-square minimization \eqref{eq:orthogonalprojections} as \begin{equation}
%\alpha^{*,K}_j = E\bigg [B(\geoRP{X}_{[0,t_j]})^{\top}B(\geoRP{X}_{[0,t_j]})\bigg ]^{-1}E\bigg [B(\geoRP{X}_{[0,t_j]})Z_{\tau_{j+1}}\bigg],\quad j = 1,\dots N-1,\label{eq:minimizerLS}
%\end{equation} where $B = (b^0,\dots,b^D)^{\top}$. 
for $i=1,\dots,M$ i.i.d sample paths of $Z$ and the discretized and truncated signature $\mathbf{X}^{\leq K} = \mathbf{X}^{\leq K}(J)$, assuming that $\tau^{K,J}_{n+1}$ is known, we estimate $\ell^{*}$ by solving \eqref{eq:orthogonalprojections} via linear least-square regression. This yields an estimator $\ell^{*} =\ell^{*,n,J,K,M}$. Defining $\psi^{n,J,K,M}(\mathbf{x})= \langle \mathbf{x}^{\leq K} ,\ell^{*,n,J,K,M}\rangle$ leads to a recursive algorithm for stopping times, for $i=1,\dots,M$  \begin{equation}
    \begin{aligned}
        \tau^{K,J,(i)}_N &= t_N  \\ \tau^{K,J,(i)}_n & = t_n1_{\big\{Z^{(i)}_{t_n} \geq \psi^{n,J,K,M}(\geoRP{X}^{(i)}|_{[0,t_n]})\big \}}+\tau^{K,J,(i)}_{n+1}1_{\big\{Z^{(i)}_{t_n} < \psi^{n,J,K,M}(\geoRP{X}^{(i)}|_{[0,t_n]})\big \}}.\label{eq:recur2}
    \end{aligned}
\end{equation}  Then the following law of large number type of result holds true, which almost directly follows from Clément et al. \cite[Theorem 3.2]{clement2002analysis}, see Appendix \ref{appendixLS}.
    \begin{proposition}\label{THM:LSconvergence2} 
    
    %For all $n=0,\dots,N$ we have \begin{equation*}
        %\lim_{K\to \infty}\lim_{J\to \infty} E[Z_{\tau_{n}^{K,J}}|\mathcal{F}_{t_n}^{\mathbf{X}}] = E[Z_{\tau_n}|\mathcal{F}_{t_n}^{\mathbf{X}}]  \text{ in } L^2.
    %\end{equation*} 
    For fixed $K,J$ we have $$\lim_{M\to \infty}\frac{1}{M}\sum_{i=1}^MZ^{(i)}_{\tau^{K,J,(i)}_n}=E[Z_{\tau_n^{K,J}}] \text{ a.s. }$$ Moreover, setting $y_0^{N,K,J,M} := \max (Z_{t_0},\frac{1}{M}\sum_{i=1}^MZ^{(i)}_{\tau^{K,J,(i)}_1})$ we have \begin{equation*} \lim_{K\to \infty}\lim_{J\to \infty}\lim_{M\to \infty} y_0^{K,N,J,M}
 = y_0^{N},
\end{equation*} where the convergence with respect to $M$ is almost sure convergence.
    \end{proposition}
    
\begin{remark} The recursion of stopping times \eqref{eq:recur2}, resp. the resulting linear functionals of the signature $\psi^{n,K,M}$, provide a stopping policy for each sample path of $Z$. By resimulating $\tilde{M}$ i.i.d samples of $Z$ and the signature $\Sig{X}$, we can notice that the resulting estimator $y_0^{K,N,J,\tilde{M}}$ is lower-biased, that is $y_0^{K,N,J,\tilde{M}} \leq y_0^N$, since the latter is defined by taking the supremum over all possible stopping policies.
\end{remark}
    
\subsection{Dual optimal stopping with signatures }\label{sec:DPsignaturesection}
In this section, we approximate solutions to the optimal stopping problem in its dual formulation, leading to upper bounds $y_0^U\geq y_0$ for \eqref{eq:OSPmain}. The dual representation goes back to Rogers \cite{rogers2002monte}, where the author shows that under the assumption $\sup_{0\leq t \leq T}|Z_t| \in L^2$, the optimal stopping problem \eqref{eq:OSPmain} is equivalent to \begin{equation}
    y_0 = \inf_{M \in \mathcal{M}^2_0} E \Big [\sup_{t \leq T} (Z_t - M_t ) \Big ],\label{DPsectionDP} 
\end{equation} where $\mathcal{M}_0^2$ denotes the space of $\mathbb{F}^{\mathbf{X}}-$martingales in $L^2$, starting from 0. Assuming that $\mathbb{F}^{\mathbf{X}}-$ is generated by a Brownian motion $W$, we can prove the following equivalent formulation of \eqref{DPsectionDP}. 
\begin{remark}
    From a financial modeling perspective, assuming that the filtration is generated by an $m$-dimensional Brownian motion $W$ is natural in the context of diffusive or rough volatility modeling. In this setting, $W$ represents the driver of $m/2$ assets (semimartingales) and their corresponding $m/2$ variance processes. 
\end{remark}
\begin{theorem}\label{mainresultDP}
    Assume that $\mathbb{F}^{\mathbf{X}}$ is generated by a $m$-dimensional Brownian motion $W$. Then for all $M\in \mathcal{M}_0^2$, there exist sequences $\ell^{i}=(\ell^i_n)_{n\in \mathbb{N}} \subset \mathcal{W}^{d+1}$ for $i=1,\dots,m$, such that \begin{equation*}
        \int_0^{\cdot}\langle \mathbf{X}_{0,s}^{<\infty},\ell_n\rangle ^{\top} dW_s:=\sum_{i=1}^m\int_0^{\cdot}\langle \mathbf{X}_{0,s}^{<\infty},\ell^i_n\rangle dW^i_s \rightarrow M_{\cdot} \text{ as } n \to \infty, \text{ ucp.}
    \end{equation*} In particular, the minimization problem \eqref{DPsectionDP} can be equivalently formulated as \begin{equation}
    \begin{aligned}
    y_0 & =  \inf_{\ell \in (\mathcal{W}^{d+1})^{m}} E \bigg[ \sup_{t \leq T}
    \bigg( Z_t - \int_0^t  \langle \Sig{X}_{0,s},\ell \rangle ^{\top} dW_s \bigg) \bigg] \\ & = \inf_{\ell^1,\dots,\ell^m \in \mathcal{W}^{d+1}} E \bigg[ \sup_{t \leq T}
    \bigg( Z_t - \sum_{i=1}^m\int_0^t \langle \Sig{X}_{0,s},\ell^i \rangle dW^i_s \bigg) \bigg].\label{minidpthm}
    \end{aligned}
  \end{equation}
\end{theorem}
\begin{proof}
    By the Martingale Representation Theorem, see for instance Karatzas and Shreve \cite[Theorem 4.5]{karatzas1991brownian}, any  $\mathbb{F}^{\mathbf{X}}-$martingale can be
  represented as \begin{equation*}
      M_t = \int_0^t \alpha_s ^{\top} dW_s = \sum_{i=1}^m\int_0^t\alpha^{i}_sdW^i_s,
  \end{equation*}
  where $(\alpha_s)_{s \in [0,T]}$ is $\mathbb{F}^{\mathbf{X}}-$adapted, measurable and square integrable. Moreover, since $M\in \mathcal{M}_0^2$, it follows that $E [M_T^2] = E [\int_0^T |\alpha_t|^2 dt ] < \infty$. From \cite[Proposition 1.1.12]{karatzas1991brownian}, we know that any adapted and measurable process has a progressively measurable modification, which we again denote by $\alpha$. From Corollary \ref{mainresultcoro}, we know that there exist sequences $\ell^i=(\ell^i_n)_{n\in \mathbb{N}} \subset \mathcal{W}^{d+1}$ for $i=1,\dots,m$, such that for $\alpha^{i,n}_t:= \langle \mathbf{X}^{<\infty}_{0,t},\ell_n^i \rangle$, we have $\Vert \alpha^{i,n}-\alpha^i\Vert_{\mathbb{H}^2} \longrightarrow 0$ as $n \to \infty$. Using Doobs inequality, we in particular have \begin{align*}
      E \bigg [\bigg (\sup_{t\leq T}\int_0^t (\alpha^n_s-\alpha_s )^{\top}dW_s\bigg )^2\bigg] & \lesssim\sum_{i=1}^m E\bigg [\int_0^T|\alpha^{i,n}_s-\alpha^i_s|^2dt\bigg ] \\ & = \sum_{i=1}^m \Vert \alpha^{i,n}-\alpha^i \Vert_{\mathbb{H}^2}^2 \longrightarrow 0.
  \end{align*} But this readily implies the first claim, that is \begin{equation*}
        \int_0^{\cdot}\langle \mathbf{X}_{0,s}^{<\infty},\ell_n\rangle ^{\top} dW_s \longrightarrow \int_0^{\cdot}\alpha_sdW_s = M_{\cdot} \text{ ucp.}
    \end{equation*} In order to show \eqref{minidpthm}, since $\int_0^{\cdot}\langle \mathbf{X}_{0,s}^{<\infty},\ell_n\rangle ^{\top} dW_s$ are clearly $\mathbb{F}^{\mathbf{X}}$-martingales, we can notice that \begin{align*}
        \inf_{\ell \in (\mathcal{W}^{d+1})^m} E \bigg[ \sup_{t \leq T}
    \bigg( Z_t - \int_0^t  \langle \Sig{X}_{0,s},\ell \rangle ^{\top} dW_s \bigg) \bigg]  & \geq \inf_{M \in \mathcal{M}_0^2} E \Big[ \sup_{t \leq T}
    ( Z_t - M_t) \Big] \\ & = y_0.
    \end{align*} On the other hand, for any fixed martingale $M$, we know there exist sequences of words $\ell^i=(\ell^i_n)_{n\in \mathbb{N}} \subset \mathcal{W}^{d+1}$ such that  \begin{equation*}
        \lim_{n \to \infty}\sup_{t \leq T}
    \bigg( Z_t - \int_0^t \langle \mathbf{X}^{<\infty}_{0,s},\ell_n \rangle ^{\top} dW_s\bigg ) = \sup_{t \leq T}
    (Z_t - M_t) \text{ in } L^2.
    \end{equation*} Therefore \begin{align*}
        E \Big [ \sup_{t\leq T}(Z_t-M_t)\Big] & = \lim_{n\to \infty} E\bigg[\sup_{t \leq T}
    \bigg( Z_t - \int_0^t \langle \mathbf{X}^{<\infty}_{0,s},\ell_n \rangle^{\top}dW_s\bigg ) \bigg] \\ & \geq \inf_{l \in (\mathcal{W}^{d+1})^m} E \bigg[ \sup_{t \leq T}
    \bigg( Z_t -   \int_0^t  \langle \Sig{X}_{0,s},\ell \rangle ^{\top} dW_s \bigg) \bigg].
    \end{align*} Taking the infimum over all $M\in \mathcal{M}_0^2$ yields the claim.
\end{proof} Next, similar to the primal case, we translate the minimization problem \eqref{minidpthm} into a finite-dimensional optimization problem, by discretizing the interval $[0,T]$ and truncating the signature to some level $K$. More precisely, for $0=t_0<\dots <t_N=T$ and some $K \in \mathbb{N}$, we reduce the minimization problem \eqref{minidpthm} to \begin{equation}
    y_0^{K,N} =\inf_{\ell \in (\mathcal{W}_{\leq K}^{d+1})^m}E\Big[\max_{0\leq n \leq N}(Z_{t_n}-M^{\ell}_{t_n})\Big],\label{approx1}
\end{equation} where for any $\ell=(\ell^1,\dots,\ell^m) \in  (\mathcal{W}_{\leq K}^{d+1})^m$ we define $$M^{l}_t=\int_0^t \langle \mathbf{X}^{\leq K}_{0,s},\ell \rangle ^{\top}
  dW_s = \sum_{i=1}^m\int_0^t \langle \mathbf{X}^{\leq K}_{0,s},\ell^i \rangle
  dW^i_s.$$ 
The discrete version of the dual formulation \eqref{DPsectionDP} is given by \begin{equation*}
    y_0^N =\inf_{M\in \mathcal{M}^{2,N}_0}E\Big[\max_{0\leq n \leq N} (Z_{t_n}-M_{t_n})\Big],
\end{equation*} where $\mathcal{M}^{2,N}_0$ denotes the space of discrete $\mathbb{F}^{\mathbf{X},N}-$martingales. The following result shows that the minimization problem \eqref{approx1} has a solution and the optimal value converges to $y_0^N$ as the level of the signature goes to infinity, the proof can be found in Appendix \ref{appendixSAA}.
\begin{proposition}\label{propositionapproxi}
    There exists a minimizer $l^{\star}$ to \eqref{approx1} and \begin{equation*}
    |y_0^N-y_0^{K,N}| \longrightarrow 0 \quad \text{as } K \rightarrow \infty.
    \end{equation*}
\end{proposition}

\begin{remark}
   In a financial context, Proposition \ref{THM:LSconvergence1} and \ref{propositionapproxi} tell us that $y_0^{K,N}$ converges to the Bermudda option price as $K \to \infty$. Moreover, we can use the triangle inequality to find \begin{equation*}
       |y_0-y_0^{K,N}| \leq |y_0-y_0^{N}| + |y_0^N-y_0^{K,N}|,
   \end{equation*} and hence the finite-dimensional approximations converge to $y_0$ as $K,N \to \infty$, whenever the Bermuddan price converges to the American price. For our numerical examples, we will always approximate $y_0^N$ for some fixed $N$, and therefore we do not further investigate in the latter convergence here.
\end{remark}
\subsubsection{Sample average approximation (SAA)}\label{sec:SAAsection}
We now present a method to approximate the value $y_0^{K,N}$ in \eqref{approx1}, using Monte-Carlo simulations. This procedure is called \textit{sample average approximation} (SAA) and we refer to Shapiro \cite[Chapter 6]{shapiro2003monte} for a general and extensive study of this method. Similar to the primal case, we introduce two further approximation steps: First, let $0=s_0<\dots < s_J = T$ be a finer discretization of $[0,T]$ and denote by $M^{l,J}$, resp. $\Sig{X}(J)$, an approximation of the stochastic integral $\int \langle \Sig{X}_{0,s},\ell \rangle dW_s$, resp. the signature $\Sig{X}$, using an Euler-scheme. Secondly, we consider $i=1,\dots, M$ i.i.d. sample paths $Z^{(i)},M^{(i),l,J}$, and replace the expectation in \eqref{approx1} by a sample average, leading to the following empirical minimization problem \begin{equation}
    y_0^{K, N, J, M} = \inf_{l \in  (\mathcal{W}_{\leq K}^{d+1})^m} \frac{1}{M}
    \sum_{i = 1}^M \max_{0 \leq n \leq N} ( Z^{(i)}_{t_n} -M^{(i),l,J}_{t_n}). \label{finalminimizationapprox}
  \end{equation}
The following result can be deduced from Shapiro \cite[Chapter 6 Theorem 4]{shapiro2003monte}, combined with Proposition \ref{propositionapproxi}, we refer to Appendix \ref{appendixSAA} for the details.
\begin{proposition}\label{THM:SAAapproxi2}
    For $M$ large enough there exists a minimizer $\beta^{\star}$ to \eqref{finalminimizationapprox} and \begin{equation*}
        \lim_{K \to \infty}\lim_{J \to \infty}\lim_{M\to \infty}y_0^{K,N,J,M}=y_0^N,
    \end{equation*} where the convergence with respect to $M$ is almost sure convergence.
\end{proposition}
\begin{remark}\label{linearprogramm} Let us quickly describe how we will solve \eqref{finalminimizationapprox} numerically: Consider the number $D:= \sum_{k=0}^K(d+1)^k$, which corresponds to the number of entries of the $K$-step signature. Notice that for any element $l \in \mathcal{W}^{d+1}_{\leq K}$, we have the following representation $\ell= \lambda_1w_1+ \dots +\lambda_Dw_D$, where $w_1,\dots,w_D$ are all possible words of length at most $K$. Since $\langle \mathbf{X}^{\leq K}_{0,t},\ell \rangle = \sum_{r=1}^D\lambda_r\langle \mathbf{X}^{\leq K}_{0,t},w_r \rangle$, the minimization \eqref{finalminimizationapprox} has equivalent formulation \begin{equation*}
    y_0^{K, N, J, M} = \inf_{\lambda \in  (\R^D)^m} \frac{1}{M}
    \sum_{i = 1}^M \max_{0 \leq n \leq N} ( Z^{(i)}_{t_n} -\sum_{r=1}^D\lambda_rM^{(i),w_r,J}_{t_n}).
  \end{equation*}
    As described in Desai et al. \cite{desai2012pathwise}, the latter minimization problem is equivalent to the following
  linear program
  \begin{equation*}
    \min_{x \in \mathbb{R}^{M + D}} \frac{1}{M} \sum_{j = 1}^M x_j, \quad
    \tmop{subject} \tmop{to} \tmop{Ax} \geq b,
  \end{equation*}
  where $A \in \mathbb{R }^{M (N+1)\times (M + D)}$ with $A = [A^1, \ldots,
  A^M]^T$ and $Ax \geq b$ represents the constraints \begin{equation*}
   x_i \geq Z^{(i)}_{t_n}-\sum_{r=1}^DM_{t_n}^{(i),w_r,J}, \quad \begin{aligned} & i = 1,\dots,M \\ &n = 0,\dots,N \end{aligned}.
\end{equation*}
\end{remark}
\begin{remark} A solution $l^{\star}$ to \eqref{finalminimizationapprox} yields the $\mathbb{F}^{\mathbf{X},N}-$martingale $M^{l^{\star}}$, and by resimulating $\tilde{M}$ i.i.d samples of $Z$, the Brownian motion $W$ and the signature $\Sig{X}$, we can notice that the resulting estimator $y_0^{K,N,J,\tilde{M}}$ is upper-biased, that is $y_0^{K,N,J,\tilde{M}} \geq y_0^N$, since the latter is defined by taking the infimum over all $\mathbb{F}^{\mathbf{X},N}-$martingales.
\end{remark}

\section{Numerical examples}\label{sec:Numerics}
In this section we study two non-Markovian optimal stopping problems and test our methods to approximate lower and upper bounds for the optimal stopping value. The details for the implementations and examples can be found here \url{https://github.com/lucapelizzari/Optimal_Stopping_with_signatures}.

\begin{remark}
    In all the performed numerical experiments below, we did not observe any significant difference when using the normalized signature (choosing the same normalization the authors introduce in Chevyrev and Oberhauser \cite[Section 3]{chevyrev2018signature}), and when using the standard, non-normalized signature. Since the tensor normalization increases the complexity of the algorithms, all the results presented here were obtained using the standard signature.
\end{remark}

\subsection{Optimal stopping of fractional Brownian motion}\label{sec: fBM}
We start with the task of optimally stopping a fractional Brownian motion (fBm), which represents the canonical choice of a framework leaving the Markov and semimartingale regimes. Recall that a fBm with Hurst parameter $H\in (0,1)$ is the unique, continuous Gaussian process $(X^H_t)_{t\in [0,T]}$, with \begin{align*}
E[X_t^H] & = 0, \quad \forall t \geq 0, \\ E[X_s^HX_t^H] & = \frac{1}{2} \big (|s|^{2H}+|t|^{2H}-|t-s|^{2H}\big ), \quad \forall s, t \geq 0. 
\end{align*} see for instance Friz and Hairer \cite[Chapter 9]{friz2020course} for more details. We wish to approximate the value \begin{equation}\label{eq:OSFBM}
y_0^H= \sup_{\tau \in \mathcal{S}_0}E[X_{\tau}^H], \quad H\in (0,1),
\end{equation} from below and above. This example has already been studied in Becker et al. \cite[Section 4.3]{becker2019deep} as well as in Bayer et al. \cite[Section 8.1]{bayer2021optimal}, and we compare the results below.

Since $X^H$ is one-dimensional and $\alpha$-Hölder continuous for any $\alpha < H$, its (scalar) rough-path lift is given by \begin{equation*} \bigg (1,X^H_{s,t},\frac{1}{2}(X_{s,t}^H)^2,\dots,\frac{1}{L!}(X_{s,t}^{H})^{L}\bigg ) \in \mathscr{C}^{\alpha}_g([0,T],\mathbb{R}),
\end{equation*} where $L=\lfloor \frac{1}{\alpha}\rfloor$. We can extend it to a geometric rough-path lift $\geoRP{X}^H \in \widehat{\mathscr{C}}^{\alpha}_g([0,T],\mathbb{R}^2)$ of the time-augmentation $(t,X^H_t)$, as for instance described in \cite[Example 2.4]{bayer2021optimal}. To numerically solve \eqref{eq:OSFBM}, we replace the continuous-time interval $[0,T]$ by some finite grid points $0=t_0<t_1<\dots<t_N=T$. Below we compare our results with \cite[Section 4.3]{becker2019deep}, where the authors chose $N=100$. Before doing so, an important remark about the difference of our problem formulation is in order. \begin{remark}\label{rmk:BCJ}In \cite{becker2019deep} the authors lift $X^H$ to a $100$-dimensional Markov process of the form $\widehat{X}_{t_k} = (X^H_{t_k},\dots, X^H_{t_1},0,\dots,0) \in \R^{100}$, and they consider the corresponding discrete (!) filtration $\widehat{\mathcal{F}}_k=\sigma(X^H_{t_k},\dots,X^H_{t_1}), k=0,\dots,100$. Notice that this differs from our setting, as we consider the bigger filtration $\mathcal{F}_k=\sigma(X^H_{s}:s\leq t_k)$, see Section \ref{sec: LSsection} and \ref{sec:DPsignaturesection}, that contains the whole past of $X^H$, not only the information at the past exercise-dates. Thus, in general $y_0^H$ dominates the lower-bounds from \cite{becker2019deep}, simply because our filtration contains more stopping-times. Similarly, the (very sharp) upper-bounds in \cite{becker2019deep} were obtained using a nested Monte-Carlo approach, which constructs $(\widehat{\mathcal{F}}_k)$-martingales that are not martingales in our filtration, and thus their upper-bounds are not necessarily upper-bounds for \eqref{eq:OSFBM}. 
\end{remark} %\color{red} TODO: Add comment about simulation of fBM, and Anderson-Broadie algorithm with conditonal sampling according to NutzmannPoor, and adjust upper-bound tables as follows: first table $N=N1=100$ and value $y_0^{LS},y_0^{SAA},y_0^{AB}$ and $[y_0^{BCJ,down},y_0^{BCJ,up}$, and the same for $N>N1=100$, and commenting again on the difference. \color{black}
%\begin{figure}[h]
%\centering%\includegraphics[width=0.70\textwidth]{Figure_1.png} 
%\caption{Optimal stopping of fBm $H \mapsto y_0^H$, lower-bounds from Longstaff-Schwartz with signatures Section \ref{sec: LSsection}, upper-bounds from dual approach in Section \ref{sec:DPsignaturesection}, reference values from \cite{bayer2021optimal}.}
%\end{figure}
\begin{table}\centering
\ra{1.3}
\begin{tabular}{@{}rrrrc@{}}\toprule 
$H$ &  $J = 100$ & $J = 500$ & Becker et al. \cite{becker2019deep} \\ \midrule
0.01 & [1.518,1.645] & [1.545,1.631]  & [1.517,1.52] \\
0.05 & [1.293,1.396] &[1.318,1.382]   & [1.292,1.294] \\
0.1 &[1.045,1.129]& [1.065,1.117]&  [1.048,1.05]\\
0.15 &  [0.83,0.901] &[0.847,0.895]  & [0.838,0.84] \\
0.2 & [0.654,0.706] & [0.663,0.698] & [0.657,0.659]\\
0.25 &[0.507,0.538] &[0.510,0.533]  & [0.501,0.505]\\
0.3 & [0.363,0.396] &[0.371,0.392] &   [0.368,0.371] \\
0.35 & [0.248,0.272] & [0.255,0.270] &[0.254,0.257]\\
0.4 &  [0.153,0.168] & [0.155,0.165]  & [0.154,0.158]\\
0.45 & [0.069,0.077] &[0.068,0.076]   & [0.066,0.075] \\
0.5 & [-0.001,0]&  [-0.002,0] &   [0,0.005]\\
0.55 & [0.061,0.071]& [0.060,0.066] & [0.057,0.065]\\
0.6 & [0.112,0.133] & [0.112,0.124]  & [0.115,0.119] \\
0.65 & [0.163,0.187]& [0.163,0.175]& [0.163,0.166]\\
0.7 & [0.203,0.234]& [0.205,0.220]& [0.206,0.208]\\
0.75 & [0.242,0.273] & [0.240,0.260]   & [0,242,0.245]\\
0.8 & [0.275,0.306] & [0.281,0.298]&  [0.276,0.279] \\
0.85 &[0.306,0.335]& [0.301,0.324] & [0.307,0.31]\\ 
0.9 & [0.331,0.357]& [0.337,0.356]&  [0.335,0.339]\\
0.95 & [0.367,0.381]&  [0.366,0.381]&  [0.365,0.367] \\ \bottomrule

\end{tabular}
\caption{Intervals for optimal stopping of fBm $H \mapsto y_0^H$ with $N=100$ exercise-dates, and discretization $J=100$ (left column), $J=500$ (middle column), and intervals from \cite{becker2019deep} (right column). Overall Monte-Carlo error is below $0.003$.}\label{FBMtable}
\end{table}
In Table \ref{FBMtable} we present intervals for the optimal stopping values $y_0^H$ for Hurst parameters $H\in \{0.01,0.05,\dots,0.95 \}$, where the lower-bounds, resp. the upper-bounds, were approximated using Longstaff-Schwartz with signatures, resp. the SAA approach described in Section \ref{sec: OptimalStoppingSection}. We truncate the signature at level  $K=6$, and apply the primal approach using $M=10^6$ samples for both the regression and the resimulation, and for the dual approach we choose $M=15000$ to solve the linear programm from Remark \ref{linearprogramm}, and resimulate with $M = 10^5$ samples. In the first column, we choose the time-discretization for the signature equal to the number of exercise-dates, by $J=N=100$. While the lower-bounds are very close, our upper-bounds exceed the ones from \cite{becker2019deep}. This observation matches with the comments made in Remark \ref{rmk:BCJ}, as we consider the filtration $\widehat{\mathcal{F}}$ in this case for the lower-bounds, but our upper-bounds are by construction upper-bounds for the continuous problem with filtration $\mathcal{F}$, and the continuous martingale is approximated only at the exercise-dates. By increasing the discretization to $J=500$, and thereby adding information to the filtration in-between exercise-dates, for small $H$ ($\leq 0.2$), one can see that the lower-bounds exceed the intervals from \cite{becker2019deep}, showing that even for $N=100$ points in $[0,1]$, the information in-between exercise-dates is relevant for optimally stopping the fBm. 

\subsection{American options in rough volatility models}\label{sec:AmericanOptions}

The second example we present is the problem of pricing American options in rough volatility models. More precisely, we consider the one-dimensional asset-price model \begin{equation*}
X_0=x_0, \quad dX_t = rX_tdt+X_tv_t (\rho dW_r+\sqrt{1-\rho^2}dB_t ), \quad 0 < t \leq T,
\end{equation*} where $W$ and $B$ are two independent Brownian motions, the volatility $(v_t)_{t\in[0,T]}$ is an $\mathbb{F}^W-$adapted, continuous process, $\rho \in [-1,1]$ and $r>0$ the interest rate. Now for any payoff function $\phi: [0,T] \times \mathbb{R} \rightarrow \mathbb{R}$, we want to approximate the optimal stopping problem \begin{equation}\label{eq:AmericanOption}
y_0 = \sup_{\tau \in \mathcal{S}_0}E[e^{-r\tau}\phi(\tau,X_\tau)],
\end{equation} where $\mathcal{S}_0$ is the set of $\mathbb{F}:=\mathbb{F}^W \lor \mathbb{F}^B-$stopping times on $[0,T]$. It is worth to note that our method does not depend on the specification of $v$, and as soon as we can sample from $(X,v)$, we can apply it to approximate values of American options.

In the following we focus on the rough Bergomi model Bayer et al. \cite{bayer2016pricing}, that is we specify the volatility as \begin{equation*}
    v_t=\xi_0\mathcal{E}\bigg (\eta \int_0^t(t-s)^{H-\frac{1}{2}}dW_s \bigg ),
\end{equation*} where $\mathcal{E}$ denotes the stochastic exponential, and we will consider the parameters $x_0 = 100, r=0.05, \eta = 1.9, \rho = -0.9, \xi_0= 0.09$. Moreover, we consider put options $\phi(t,x) = (K-x)^{+}$ for different strikes $K\in \{70,80,\dots,120\}$, with maturity $T=1$ and $N=12$ exercise-dates. Thus we write \eqref{eq:AmericanOption} as discrete optimal stopping problem \begin{equation*}
y_0^{N}=\sup_{\tau \in \mathcal{S}^N_0}E \big [e^{-r\tau} (K-X_{\tau} )^{+}\big],
\end{equation*} where $\mathcal{S}_0^N$ is described in the beginning of Section \ref{sec: LSsection}. Moreover, for some finer grid $s_0=0<s_1 < \dots < s_J=1, J\in \N$ and fixed signature level $K$ and sample size $M$, we denote by $y_0^{\text{LS}}$ the value $y_0^{K,N,J,M}$ defined in Proposition \ref{THM:LSconvergence2}, resp. by $y_0^{\text{SAA}}$ the value $y_0^{K,N,J,M}$ defined in \eqref{finalminimizationapprox}. We compare our results with the lower-bounds obtained in Bayer et al. \cite{bayer2020pricing} for $H=0.07$, resp. in Goudenege et al. \cite{goudenege2020machine} with $H=0.07$ and $H=0.8$.

In the first numerical experiments we will again simply consider the signature of the time-augmented path $(t,X_t)$, that is  we choose the basis functions for the least square regression \eqref{eq:orthogonalprojections}, resp. for the SAA minimization problem in \eqref{approx1} to be \begin{equation*}\begin{aligned}
& \mathbf{(B_1)} \quad \{\langle \mathbf{X}_{0,t}^{<\infty},\ell \rangle : \ell \in \mathcal{W}^{2}_{\leq K} \}, \quad K\in \N
\end{aligned}\end{equation*}In the first column of Table \ref{tb:48,0.07}-\ref{tb:600,0.07} we report the price intervals $[y_0^{\text{LS}},y_0^{\text{SAA}}]$ for the Hurst parameter $H=0.07$, and different discretizations $J=48$ and $J=600$. The degree of the signature is fixed at $K=4$ , and we apply the primal algorithm described in Section \ref{sec: LSsection} for $M= 10^6$ samples. For the obtained stopping policies, we resimulate with again $M=10^6$ samples to obtain true-lower bounds $y_0^{\text{LS}}$. For the upper-bounds we solve the linear program described in Section \ref{sec:DPsignaturesection} for $M= 10^4$ samples, and then resimulate with $M=10^5$ samples to obtain true upper-bounds $y_0^{\text{SAA}}$. In the first column of Table \ref{tb:48,0.8} we consider the same problem for $H=0.8$ and $J=600$. Similar as in Section \ref{sec: fBM}, we observe that the price intervals shrink when we increase the number of discretization points between exercise-dates. However, even for $J=600$ we still observe a significant gap between lower and upper-bounds. It was already observed in Markovian frameworks, that the approximation of the Doob-martingale usually requires a careful and specific choice of basis functions (e.g. in Belomestny et al. \cite{belomestny2009true} the authors use European deltas). This motivates us to extend the basis $(\mathbf{B_1})$ in two ways: First, we add Laguerre polynomials of the states $(X_t,v_t)$, which would be a natural choice in a Markovian framework, for instance used in the original work by Longstaff and Schwartz in \cite{longstaff2001valuing}. Additionally, for the dual-problem we add the payoff process $Z_t$ to the path, that is we lift $(t,X_t,Z_t)$ to the signature $\mathbf{Z}^{<\infty}$. Since $(t,X_t,Z_t)$ is a semimartingale, the signature $\Sig{Z}$ is given as the sequence of iterated Stratonovich integrals as explained in Friz et al. \cite{friz2022unified}. Of course adding basis functions in the primal and dual approach does not change the convergence. To summarize, for the least squares regression \eqref{eq:orthogonalprojections}, resp. for the SAA minimization problem in \eqref{approx1}, we use the extended basis $(\mathbf{B}_2) =\{ (\mathbf{P}_2),(\mathbf{D}_2)\}$ of the form \begin{equation*}
\begin{aligned}
& \mathbf{(P_2)} \quad \big \{L_i(X_t,v_t),\langle \mathbf{X}_{0,t}^{<\infty},\ell \rangle : i = 1,\dots,m_p, \ell \in \mathcal{W}^{d+1}_{\leq K_p}\big\}, \quad m_p,K_p\in \N,\\ & \mathbf{(D_2)} \quad \big \{L_i(X_t,v_t),\langle \mathbf{Z}_{0,t}^{<\infty},\ell \rangle: i = 1,\dots,m_d, \ell \in \mathcal{W}^{d+1}_{\leq K_d}\big\}, \quad m_d,K_d\in \N,
\end{aligned}
\end{equation*} where $(L_k)_{k\geq 0}$ are Laguerre polynomials. In the second columns of Table \ref{tb:48,0.07}-\ref{tb:48,0.8}  we report the price-intervals for the extended Basis $\mathbf{(B_2)}$. We consider polynomials degree $3$ for the primal, and $5$ for the dual approach, and the signature levels and number of samples are the same as before. Especially for $J=600$ we observe a significant reduction of the upper-bounds, which are now only $2\%-3\%$ higher than the lower-bounds. We expect these margins to shrink more when either further increasing all parameters or when choosing non-linear basis functions, such as for instance deep neural networks on the log-signature or signature-kernel based methods. These two ideas are currently under development and more details will appear in future works.

\begin{table*}\centering
\ra{1.3}
\begin{tabular}{@{}rrrrrc@{}}\toprule 
$K$ & \text{$(\mathbf{B_1})$-Basis} & \text{$(\mathbf{B_2})$-Basis} & \textit{\cite{bayer2020pricing}} & \textit{ \cite{goudenege2020machine}} \\ \midrule
   $70$ & (1.83, 2.59) & (1.85, 2.04)  & 1.88& 1.88 & \\ $80$ & (3.18, 4.44) & (3.18, 3.44)  & 3.22 & 3.25 & \\ $90$ & (5.19, 7.66) & (5.25, 5.68) & 5.30& 5.34 & \\$100$ & (8.33, 13.16) & (8.44, 9.13) & 8.50& 8.53 & \\$110$ & (13.02, 21.38) & (13.18, 14.18) & 13.23& 13.28 & \\$120$ & (20.20, 30.91) & (20.22, 21.40) & 20& 20.20 & \\\bottomrule 
\end{tabular}
\caption{Put option price intervals with $J = 48$ and $H=0.07$, for the two choices of Basis functions described in $(\mathbf{B_1})$, resp. $(\mathbf{B_2})$, and lower-bound reference values.}\label{tb:48,0.07}
\end{table*}

\begin{table*}\centering
\ra{1.3}
\begin{tabular}{@{}rrrrrc@{}}\toprule 
$K$ & \text{$(\mathbf{B_1})$-Basis} & \text{$(\mathbf{B_2})$-Basis} & \textit{\cite{bayer2020pricing}} & \textit{ \cite{goudenege2020machine}} \\ \midrule
   $70$ & (1.90, 2.38) & (1.92, 1.99)  & 1.88& 1.88 & \\ $80$ & (3.25, 4.13) & (3.27, 3.37) & 3.22 & 3.25 & \\ $90$ & (5.34, 7.17) &  (5.37, 5.49) & 5.30& 5.34 & \\$100$ & (8.51, 12.55) &  (8.57, 8.77) & 8.50& 8.53 & \\$110$ & (13.24, 20.79) &  (13.29, 13.59) & 13.23& 13.28 & \\$120$ & (20.22, 29.90) &  (20.24, 20.66) & 20& 20.20 & \\\bottomrule 
\end{tabular}
\caption{Put option price intervals with $J = 600$ and $H=0.07$, for the two choices of Basis functions described in $(\mathbf{B_1})$, resp. $(\mathbf{B_2})$, and lower-bound reference values. }\label{tb:600,0.07}
\end{table*}

\begin{table*}\centering
\ra{1.3}
\begin{tabular}{@{}rrrrc@{}}\toprule 
$K$ & \text{$(\mathbf{B_1})$-Basis} & \text{$(\mathbf{B_2})$-Basis} & \textit{\cite{goudenege2020machine}} \\ \midrule
   $70$ &  (1.83, 2.39)& (1.83, 1.90) & 1.84& \\ $80$ &  (3.08, 4.13)& (3.08, 3.19)  &  3.10 & \\ $90$ & (5.04, 7.38) & (5.07, 5.17) & 5.08 & \\$100$ &  (8.11, 12.84) & (8.15, 8.27)  & 8.19 & \\$110$ &  (12.89, 20.77) & (12.97, 13.09)  & 13.00 & \\$120$ &  (20.16, 30.21)& (20.21, 20.51) & 20.28 & \\\bottomrule 
\end{tabular}
\caption{Put option price intervals with $J = 600$ and $H=0.8$, for the two choices of Basis functions described in $\mathbf{P_1,D_1}$, resp. $\mathbf{P_2,D_2}$, and lower-bound reference value. }\label{tb:48,0.8}
\end{table*}

%\begin{acknowledgements}
%\end{acknowledgements}



\appendix\normalsize

\section{Technical details Section \ref{sec: OptimalStoppingSection}}
\subsection{Proofs in Section \ref{sec: LSsection}}\label{appendixLS}
\begin{proof}{\textbf{of Proposition \ref{THM:LSconvergence1}}}
     The proof is based on the same ideas as the proof in \cite[Theorem 3.1]{clement2002analysis}. We can proceed by induction over $n$. For $n=N$ the claim trivially holds true, and assume it holds for $0 \leq n+1\leq N-1$. Let us recall from \eqref{eq:orthogonalprojections} that for fixed $K$ and $n$, we define $\psi^{n,K}(\mathbf{x})=\langle \mathbf{x}^{\leq K},l^{\star,n,K} \rangle$ where $$l^*:=l^{*,n,K} = \underset{l \in \mathcal{W}_{\leq K}^{d+1}}{\operatorname{argmin}} \left \Vert Z_{\tau^K_{n+1}}-\langle \mathbf{X}_{0,t_n}^{\leq K},l \rangle \right \Vert_{L^2}.$$ In particular, by the Hilbert Projection Theorem, $\psi^{n,K}(\mathbf{X}|_{[0,t]})$ is the orthogonal projection of $Z_{\tau^K_{n+1}}$ onto the subspace $\{\langle \mathbf{X}^{\leq K}_{0,t_n},l \rangle : l\in \mathcal{W}^{d+1}\}$ of $L^2$. Define the events $$A(n) := \left \{ Z_{t_n} \geq E\left [Z_{\tau_{n+1}}|\mathcal{F}^{\geoRP{X}}_{t_n}\right ] \right \} \quad \text{and} \quad A(n,K) := \left \{ Z_{t_n} \geq \psi^{n,K}(\geoRP{X}|_{[0,t_n]})  \right \}.$$ By definition we can write $$\tau_n^K = t_n1_{A(n,K)}+\tau_{n+1}^K1_{A(n,K)^C}, \quad \tau_n = t_n1_{A(n)}+\tau_{n+1}1_{A(n)^C}.$$ Using this, it is possible to check that \begin{align*}
     E[Z_{\tau_n^K}-Z_{\tau_n}|\mathcal{F}^{\geoRP{X}}_{t_n}] = (Z_{t_n}-& E[Z_{\tau_{n+1}}|\mathcal{F}^{\geoRP{X}}_{t_n}])(1_{A(n,K)}-1_{A(n)}) \\ & + E[Z_{\tau^K_{n+1}}-Z_{\tau_{n+1}}|\mathcal{F}^{\geoRP{X}}_{t_n}]1_{A(n,K)^C}.
     \end{align*} The second term converges by induction hypothesis, and we only need to show $$L_n^K := (Z_{t_n}- E[Z_{\tau_{n+1}}|\mathcal{F}^{\geoRP{X}}_{t_n}])(1_{A(n,K)}-1_{A(n)}) \xrightarrow{K \to \infty} 0, \text{ in } L^2.$$ Now on $A(n,K) \cap A(n)$ and $A(n,K)^c \cap A(n)^c$ we clearly have $L_n^K=0$. Moreover $$1_{A(n,K)^c\cap A(n)}|L_n^K|  \leq 1_{A(n,K)^c\cap A(n)}|\psi^{n,K}(\geoRP{X}|_{[0,t_n]})-E[Z_{\tau_{n+1}}|\mathcal{F}^{\geoRP{X}}_{t_n}]|,$$ since $\psi^{n,K}(\geoRP{X}|_{[0,t_n]})>Z_{t_n} \geq E[Z_{\tau_{n+1}}|\mathcal{F}^{\geoRP{X}}_{t_n}]$ on $A(n,K)^c\cap A(n)$. Similarly, one can show $$1_{A(n,K)\cap A(n)^c}|L_n^K|  \leq 1_{A(n,K)\cap A(n)^c}|\psi^{n,K}(\geoRP{X}|_{[0,t_n]})-E[Z_{\tau_{n+1}}|\mathcal{F}^{\geoRP{X}}_{t_n}]|,$$ and thus \begin{equation}\label{eq:distanceAppendix}
     |L_n^K|  \leq |\psi^{n,K}(\geoRP{X}|_{[0,t_n]})-E[Z_{\tau_{n+1}}|\mathcal{F}^{\geoRP{X}}_{t_n}]|.
     \end{equation} As mentioned above, $\psi^{n,K}$ is the orthogonal projection of the $L^2$ random variable $Z_{\tau^K_{n+1}}$ onto the subspace $\{\langle \mathbf{X}^{\leq K}_{0,t_n},l \rangle : l\in \mathcal{W}^{d+1}\}$, and similarly denote by $\hat{\psi}^{n,K}$ the orthogonal projection of $Z_{\tau_{n+1}}$ to the same space. Then we have \begin{align*}
     \Vert L_n^K \Vert_{L^2} & \leq \Vert \psi^{n,K}(\geoRP{X}|_{[0,t_n]})-\hat{\psi}^{n,K}(\geoRP{X}|_{[0,t_n]}) \Vert_{L^2}+\Vert \hat{\psi}^{n,K}(\geoRP{X}|_{[0,t_n]})- E[Z_{\tau_{n+1}}|\mathcal{F}^{\geoRP{X}}_{t_n}] \Vert_{L^2} \\ & \leq \Vert E[Z_{\tau^K_{n+1}}|\mathcal{F}^{\geoRP{X}}_{t_n}] -E[Z_{\tau_{n+1}}|\mathcal{F}^{\geoRP{X}}_{t_n}]  \Vert_{L^2}+\Vert \hat{\psi}^{n,K}(\geoRP{X}|_{[0,t_n]})- E[Z_{\tau_{n+1}}|\mathcal{F}^{\geoRP{X}}_{t_n}] \Vert_{L^2}.
     \end{align*} Now the first term converges by the induction hypothesis. For the second term, the conditional expectation of the $L^2$ random variable $Z_{\tau_{n+1}}$ is nothing else than the orthogonal projection onto the space $L^2(\mathcal{F}^{\geoRP{X}}_{t_n})$. But by Theorem \ref{mainresultappro}, for any $\epsilon > 0$ we can find $\phi \in L_{\text{Sig}}^{\lambda}$, such that $\Vert \phi(\geoRP{X}|_{[0,t_n]})-Z_{\tau_{n+1}} \Vert_{L^2} \leq \epsilon$. For $K$ large enough we have $\phi(\mathbf{X}|_{[0,t_n]}) \in \{\langle \mathbf{X}^{\leq K}_{0,t_n},l \rangle : l\in \mathcal{W}^{d+1}\}$, and thus $$\Vert \hat{\psi}^{n,K}(\geoRP{X}|_{[0,t_n]})- E[Z_{\tau_{n+1}}|\mathcal{F}^{\geoRP{X}}_{t_n}] \Vert_{L^2} \leq \Vert \hat{\psi}^{n,K}(\geoRP{X}|_{[0,t_n]})-Z_{\tau_{n+1}}\Vert_{L^2} \leq  \epsilon, $$ since $\hat{\psi}^{n,K}$ is such that the distance is minimal.
\end{proof} 

\begin{proof}{\textbf{of Proposition \ref{THM:LSconvergence2}}}
    First, we can consider the sequence of stopping times $(\tau_{n}^{K,J})$ as defined in \eqref{eq:recur2}. One can then rewrite exactly the same proof of Proposition \ref{THM:LSconvergence1} for $Z_{\tau^{K,J}_n}$, writing $ \psi^{n,K,J}$  for the orthogonal projection of $Z_{\tau_{n+1}^{K,J}}$ onto the subspace $\{\langle \mathbf{X}^{\leq K}_{0,t_n}(J),l \rangle : l\in \mathcal{W}^{d+1}\}$. Instead of equation \eqref{eq:distanceAppendix}, we have  \begin{align*}
    L_n^{K,J} & \leq |\psi^{n,K,J}(\geoRP{X}|_{[0,t_n]})-E[Z_{\tau_{n+1}}|\mathcal{F}^{\geoRP{X}}_{t_n}]| \\ & \leq |\psi^{n,K}(\geoRP{X}|_{[0,t_n]})-E[Z_{\tau_{n+1}}|\mathcal{F}^{\geoRP{X}}_{t_n}]| + |\psi^{n,K,J}(\geoRP{X}|_{[0,t_n]})-\psi^{n,K}(\geoRP{X}|_{[0,t_n]})|, 
\end{align*} where the first term converges in $L^2$ due to the same argument as in the proof of Proposition \ref{THM:LSconvergence1}. Then, since we assume that $\langle \Sig{X}_{0,t}(J),v \rangle  \xrightarrow{J \to \infty} \langle \Sig{X}_{0,t},v \rangle$ in $L^2$ for all words $v$, for any fixed $K$ the second term converges in $L^2$. Thus it follows that for all $n=0,\dots, N$ we have \begin{equation}\label{eq:L2convergence}\lim_{K\to \infty}\lim_{J\to \infty}E[Z_{\tau_n^{K,J}}|\mathcal{F}^{\mathbf{X}}_{t_n}]= E[Z_{\tau_n}|\mathcal{F}_{t_n}^{\mathbf{X}}] \text{ in } L^2.\end{equation} Next, we want to show that \begin{equation}\label{eq:asconvergence}\frac{1}{M}\sum_{i=1}^MZ^{(i)}_{\tau^{K,J,(i)}_n} \xrightarrow{M\to \infty} E[Z_{\tau^{K,J}_n}]  \text{ a.s. }\end{equation} Now for any $l\in \mathcal{W}^{d+1}$, we can write $l = \lambda_1w_1+ \dots +\lambda_Dw_D$, where $D = \sum_{k=0}^K(d+1)^k$, that is we sum over all possible words of length at most $K$. One can therefore notice that minimizing $\langle \mathbf{X}^{\leq K},l \rangle$ over $l\in \mathcal{W}^{d+1}_{\leq K}$, is equivalent to minimizing $\sum_{i=1}^D\lambda_i\langle \mathbf{X}^{\leq K},w_i \rangle$ over all vectors $\lambda \in \R^D$. Defining $e_k(\mathbf{x}):=\langle \mathbf{x}^{\leq K},w_k \rangle$ for $k=1,\dots,D$, and setting $X_n:=\mathbf{X}|_{[0,t_n]}$, we are exactly in framework of \cite[Chapter 3]{clement2002analysis}, and the result follows from \cite[Theorem 3.2]{clement2002analysis}, under the following remark. The authors make the following assumption, denoted by (A2) \begin{equation}\label{linearindep}
\sum_{j}\alpha_je_j(X_t) = 0 \text{ almost surely implies } \alpha =0, \forall t
\end{equation}for the set of basis-functions, which allows an explicit representation of the coefficient $l^{\star}$ in \eqref{eq:orthogonalprojections}. Of course, in our framework, such an assumption cannot hold true, as this would correspond to $$\sum_{l=1}^D\alpha_l\langle \Sig{X}_{0,t},w_l \rangle = 0 \text{ a.s. } \Longrightarrow \alpha_l=0, \forall l =1,\dots,D.$$ Since we consider the signature of the time-augmented path $(t,X_t)$, the purely deterministic components of the signature contradict this assumption. However, for a fixed signature level $K$, we can always discard linear-dependent (in the sense of \eqref{linearindep}) components of the signature, that is minimize over the basis-functions $$\{\tilde{e}_1,\dots,\tilde{e}_{\tilde{D}}\} \subset \{e_1,\dots, e_D\} \text{ s.t. (A2) holds },$$ for the largest possible $\tilde{D} \leq D$. The resulting least-square problem \eqref{eq:orthogonalprojections} over $\R^{\tilde{D}}$, with respect to $\{\tilde{e}_1,\dots,\tilde{e}_{\tilde{D}}\}$, has an explicit representation of the solution, and since the two sets of basis-functions generate the same subspace of $L^2$, the explicit solution is also optimal for the original problem. Thus, for a fixed level $K$, we can proceed with the reduced set of basis-functions, for which the assumption (A2) holds by definition, and we can apply \cite[Theorem 3.2]{clement2002analysis}. Finally, the last convergence directly follows by combining \eqref{eq:L2convergence} and \eqref{eq:asconvergence}.
\end{proof} 
\subsection{Proofs in Section \ref{sec:DPsignaturesection}}
\label{appendixSAA}
\begin{proof}{\textbf{of Proposition \ref{propositionapproxi}}}
 The existence of a minimizer is proved in Lemma \ref{existencelemma}. We can find the discrete Doob-martingale $M^{\star,N}$ and write \begin{equation*}
        y_0^N = E\left[\max_{0\leq n \leq N}\left(Z_{t_n}-M^{\star,N}_{t_n}\right)\right].\end{equation*} Define the continuous-time, $(\mathcal{F}^{\mathbf{X}}_t)-$martingale $M_t:=E[M^{\star,N}_T|\mathcal{F}^{\mathbf{X}}_t]$, and notice that $M_{t_n}=M^{\star,N}_{t_n}$. Let us recall the following notation introduced in Section \ref{sec:DPsignaturesection}: for any $l\in (\mathcal{W}^{d+1})^m$, we define the martingale $M^l$ to be $$M^l_t = \int_0^t\langle \mathbf{X}^{<\infty}_{0,s},l\rangle ^{\top}dW_s=\sum_{i=1}^{m}\int_0^t\langle \mathbf{X}^{<\infty}_{0,s},l^i\rangle dW^i_s.$$ An application of the martingale approximation in Theorem \ref{mainresultDP} shows that for all $\epsilon >0$, there exist an $l^{\epsilon}=(l^{i,\epsilon})_{i=1}^m$ in $\left (\mathcal{W}^{d+1}\right)^m$, such that $$E\left [\max_{0\leq n \leq N}\left (M_{t_n}^{\star,N}-M_{t_n}^{l^{\epsilon}}\right )\right ]\leq \epsilon.$$ Thus we have\begin{equation*}
            y_0^N= E\left[\max_{0\leq n \leq N}\left(Z_{t_n}-M^{\star,N}_{t_n}\right)\right] \geq E\left[\max_{0\leq n \leq N}\left(Z_{t_n}-M^{l^{\epsilon}}_{t_n}\right)\right]-\epsilon.
        \end{equation*} Now since $y_0^{K,N} \geq y_0^{N}$, we can find $K$ large enough, such that \begin{align*}
            0 \leq y_0^{K,N}-y_0^N & \leq \inf_{l \in \left (\mathcal{W}_{\leq K}^{d+1}\right)^m}E\left[\max_{0\leq n \leq N}\left(Z_{t_n}-M^{l}_{t_n}\right)\right]-E\left[\max_{0\leq n \leq N}\left(Z_{t_n}-M^{l^{\epsilon}}_{t_n}\right)\right]+\epsilon \\ & \leq \epsilon,
        \end{align*} where the last inequality follows from that fact that $l^{\epsilon}\in \left (\mathcal{W}_{\leq K}^{d+1}\right)^m$ for $K$ large enough.
\end{proof}
\begin{lemma}\label{existencelemma}
    The minimization problem \begin{equation*}
    y_0^{K,N} = \inf_{l \in \left (\mathcal{W}_{\leq K}^{d+1}\right)^m}E\left[\max_{0\leq n\leq N}\left(Z_{t_n}-M^{l}_{t_n}\right)\right]
\end{equation*} has a solution.
\end{lemma}
\begin{proof} First notice that $l \mapsto E\left [\max_{0\leq n\leq N}\left(Z_{t_n}-M_{t_n}^{l}\right)\right ]$ is convex. Then we have \begin{align*}
        E\left [\max_{0\leq n\leq N}\left(Z_{t_n}-M_{t_n}^{l}\right)\right ] & \geq E\left[\max\left(Z_T-M_{T}^{l},0\right)\right] \\ &=\frac{1}{2}E\left[Z_T-M_{T}^{l} + \left|M_{T}^{l}-Z_T\right|\right] \\ & \geq \frac{1}{2}E\left[\left|M_{T}^{l}\right|\right] + E[\max(-Z_T,0)],
    \end{align*}where the equality in the middle uses $\max(A-B,0) = \frac{1}{2}\left(A-B + |B-A|\right)$. Now for any word $l=\lambda_1w_1 + \dots + \lambda_nw_n$, we set $|l| = \sum_{i=1}^n|\lambda_i|$, and notice that \begin{equation}
        \frac{1}{2}E\left[\left|M_{T}^{\l}\right|\right] = \frac{1}{2}|l|E[|M_{T}^{l/|l|}|] \geq \frac{|l|}{2}\inf_{\widehat{l} \in \left (\mathcal{W}_{\leq K}^{d+1}\right)^m,|\widehat{l}|=1}E[|M^{\widehat{l}}_{T}|].\label{infi}
    \end{equation} Since $\widehat{l} \mapsto E[|M^{\widehat{l}}_{T}|]$ is continuous and the set $\{\widehat{l} \in \left (\mathcal{W}_{\leq K}^{d+1}\right)^m: |\widehat{l}|=1\}$ is compact, the minimum on the right hand-side of \eqref{infi} is attained. Assume now that $\inf_{\widehat{l} \in \left (\mathcal{W}_{\leq K}^{d+1}\right)^m,|\widehat{l}|=1}E[|M^{\widehat{l}}_{T}|]=0$. Then there exists an $\widehat{l}^{\star}$ with $|\widehat{l}^{\star}|=1$ and $|M_T^{\widehat{l}^{\star}}|=0$ almost surely. But notice that $M^{\widehat{l}^{\star}}$ is a true martingale with quadratic variation given by $[M^{\widehat{l}^{\star}}]_T=\sum_{i=1}^m\int_0^T\langle \mathbf{X}^{\leq K}_{0,s},\widehat{l}^{\star,i} \rangle^2 ds$. Since in particular $(M^{\widehat{l}^{\star}})^2 =0$ almost surely, the same is true for the quadratic variation, and hence in particular for each term it holds that $\int_0^T\langle \mathbf{X}^{\leq K}_{0,s},\widehat{l}^{\star,i} \rangle^2 ds=0$ almost surely. But this implies that for all $i=1,\dots,m$ \begin{equation*}
        \langle \mathbf{X}^{\leq K}_{0,s},\widehat{l}^{\star,i} \rangle = 0,\quad \text{ for almost every } s\in [0,T] \text{ almost surely.}
    \end{equation*} But this is only possible if $\widehat{l}^{\star}=0$, contradicting the fact that $|\widehat{l}^{\star}|=1$. Hence the infimum \eqref{infi} is positive and we can conclude that the function $$l \mapsto E\left [\max_{0\leq n\leq N}\left(Z_{t_n}-M_{t_n}^{l}\right)\right ] \xrightarrow{|l| \rightarrow \infty} \infty,$$ which implies the existence of the minimizer.
\end{proof}
Finally, in order to prove Proposition \ref{THM:SAAapproxi2}, we quickly introduce the general idea of sample average approximation (SAA), for which we refer to \cite[Chapter 6]{shapiro2003monte} for details. Assume $\mathcal{X}$ is a closed and convex subset of $\R^N$ and $\xi$ is a random vector in $\mathbb{R}^d$ for some $d,N\in \N$, and $F$ is some function $F:\mathbb{R}^N\times \mathbb{R}^d \rightarrow \mathbb{R}$. We are interested in approximating the stochastic programming problem \begin{equation}\label{eq:stochasticprogramming}
y_0 = \min_{x\in \mathcal{X}}E[F(x,\xi)].
\end{equation} Define the sample average function $F^M(x) = \frac{1}{M}\sum_{j=1}^MF(x,\xi^{j})$, where $\xi^{j},j=1,\dots,M$ are i.i.d samples of the random vector $\xi$. The sample average approximation of $y_0$ is then given by \begin{equation} \label{eq:SAAdef}
y_0^M = \min_{x\in \mathcal{X}} F^M(x).
\end{equation} The following result provides sufficient conditions for the convergence $y_0^M \xrightarrow{M \to \infty} y_0$, and a more general version can be found in \cite[Chapter 6 Theorem 4]{shapiro2003monte}. \begin{theorem}\label{thm:shapiro}
Suppose that \begin{itemize}
\item[(1)] $F$ is measurable and $x \mapsto F(x,\xi)$ is lower semicontinuous for all $\xi \in \R^d$,
\item[(2)] $x \mapsto F(x,\xi)$ is convex for almost every $\xi$,
\item[(3)] $\mathcal{X}$ is closed and convex,
\item[(4)] $f(x):=E[F(x,\xi)]$ is lower semicontinuous and $f(x)<\infty$ for all $x\in \mathcal{X}$,
\item[(5)]the set $S$ of optimal solutions to \eqref{eq:stochasticprogramming} is non-empty and bounded.
\end{itemize} Then $y_0^M \xrightarrow{M \to \infty} y_0$.
\end{theorem} Similar as in the proof of Proposition \ref{THM:LSconvergence2}, for a fixed truncation level $K \in \mathbb{N}$, we denote by $D$ the number of components of the truncated signature, that is $D = \sum_{k=0}^K(d+1)^k$. Now every $l\in \mathcal{W}^{d+1}_{\leq K}$  is of the form $l= \sum_{i=1}^D\beta_iw_i$ for some $\beta \in \mathbb{R}^D$, where $w_1,\dots,w_D$ are all words of length at most $K$. In particular, for every $l\in (\mathcal{W}^{d+1}_{\leq K})^m$, there is an $\beta \in \mathbb{R}^{m \times D}$ and our usual notation reads \begin{equation} \label{eq:MGcoeff}M^l_t =\int_0^t\langle \mathbf{X}^{\leq K}_{0,s},l\rangle ^{\top}dW_s = \sum_{j=1}^D(\beta^{j})^{\top}\int_0^t\langle \mathbf{X}^{\leq K}_{0,s},w_j\rangle dW_s,\end{equation} and instead of minimizing over $l$, we can minimize over $\beta$. Let us now formulate the minimization problem in Proposition \ref{THM:SAAapproxi2} in the language of Theorem \ref{thm:shapiro}: \begin{align*}
    E\left[\max_{0\leq n \leq N}\left(Z_{t_n}-M^{l}_{t_n}\right)\right] & =E\left[\max_{0\leq n \leq N}\left(Z_{t_n}-\sum_{i=1}^m\sum_{j=1}^D\beta^{ij}\int_0^t\langle \mathbf{X}^{\leq K}_{0,s},w_j\rangle dW^i_s\right)\right] \\ &= E\left[\max_{0\leq n \leq N}\left(Z_{t_n}-\mathbf{\beta}^{\top}\mathbf{M}_{t_n}\right)\right],
\end{align*} where we identify $\beta \in \mathbb{R}^{m\times D} \cong \mathbb{R}^{m \cdot D}$ and $\mathbf{M}_t$ is the $m\cdot D$-dimensional vector $$\left (\int_0^t\langle \mathbf{X}^{\leq K}_{0,s},w_j\rangle dW^i_s: 1\leq i \leq m, 1 \leq j \leq D\right ).$$ Finally, defining the random vector $$\xi := \left (Z_{t_0},\mathbf{M}^{\top}_{t_0},\dots,Z_{t_N},\mathbf{M}^{\top}_{t_N}\right ) \in \R^{(D+1)\cdot (N+1)\cdot m},$$ we can notice that \begin{equation}\label{eq:discretizedapprox}
    \inf_{l \in (\mathcal{W}^{d+1}_{\leq K})^m}E\left[\max_{0\leq n \leq N}\left(Z_{t_n}-M^{l}_{t_n}\right)\right] = \min_{x \in \mathbb{R}^{m\cdot D}}E[F(x,\xi)],
\end{equation} where $$F(x,\xi) := \max_{0 \leq n \leq N}\left ( \xi_{n (m \cdot D+1)}-x^{\top}\begin{pmatrix}
    \xi_{n (m \cdot D+1)+1} \\ \vdots \\ 
    \xi_{n (m \cdot D+1)+1+m\cdot D}
\end{pmatrix}\right ).$$ Therefore, setting $\mathcal{X}= \mathbb{R}^{D\cdot m}$ and $d= (D+1)(N+1)m$, the minimization problem $y_0^{K,N,J,M}$ in Proposition \ref{THM:SAAapproxi2} can simply be written in the SAA formulation \eqref{eq:SAAdef}, where additionally the stochastic integrals in $\mathbf{M}$ are replaced by the discretized versions $\mathbf{M}^{J}$. \\
 \begin{proof}{\textbf{of Proposition \ref{THM:SAAapproxi2}}} First, it is possible to rewrite the proof of Lemma \ref{existencelemma} for $F^M$ instead of the expectation when $M$ is large enough, to show that there exists a minimizer $l^{\star}$  to \eqref{finalminimizationapprox}.  Moreover, denote by $y_0^{K,N,J}$ the minimization problem \eqref{eq:discretizedapprox}, where we replace $M^l$ be the discretized version $M^{l,J}$ as described in Section \ref{sec:SAAsection}. By the same reasoning as in Proposition \ref{propositionapproxi}, we can show that $\lim_{K\to \infty}\lim_{J\to \infty}y_0^{K,N,J} = y_0^{N}$. Now for fixed $K,J,N$, we are left with showing almost sure convergence $y_0^{K,N,J,M}\xrightarrow{M \to \infty} y_0^{K,N,J}$. But this can be deduced from Theorem \ref{thm:shapiro}, if we can show that (1)-(5) hold true for our $F$. Clearly $F$ is measurable and it is easy to see that $x \mapsto F(x,\xi)$ is continuous and convex, thus (1) and (2) readily follow. Moreover (3) holds true, and in order to show (4), set $f(x)= E[F(x,\xi)]$ and notice that for $x_1,x_2 \in \mathcal{X}$ we have \begin{align*}
     |f(x_1)-f(x_2)| \leq E\left [\max_{0\leq n \leq N}\left ( (x_1-x_2)^{\top}\mathbf{M}^{J}_{t_n})\right )\right ] \leq \Vert x_1-x_2 \Vert E\left [\max_{0 \leq n \leq N}\Vert \mathbf{M}^{J}_{t_n} \Vert\right ],
 \end{align*} where we simply applied the triangle and Cauchy-Schwarz inequalities. Since $M^{l} \in L^2$ for all $l$, and application of Doobs inequality shows that the right hand side $E\left [\max_{0 \leq n \leq N}\Vert \mathbf{M}^{J}_{t_n} \Vert\right ]<\infty$, and therefore (4) follows. Finally, non-emptyness of $S$ follows from Lemma \ref{existencelemma}, and the proof of the latter reveals that $E[F(x,\xi)] \xrightarrow{|x|\to \infty} \infty$, and thus $S$ must be bounded, which finishes the proof.
\end{proof}

\bibliographystyle{plain}
\bibliography{BIB}

\end{document}